\documentclass{article}
\usepackage[T1]{fontenc}
\usepackage[utf8]{inputenc}
\usepackage[authoryear]{natbib}
\setcitestyle{authoryear,open={(},close={)}}

\usepackage{pgfplots}
\pgfplotsset{compat=1.17}

\usepackage{amsmath}
\usepackage{amsthm}
\usepackage{amsfonts}
\usepackage{geometry}

\usepackage{comment}
\usepackage{xcolor}
\usepackage{tikz}
\usepackage{graphics}
\usepackage{graphicx}
\usepackage{color}
\usepackage[normalem]{ulem}
\usepackage{float}

\usepackage{mathabx} 

\usepackage[breaklinks]{hyperref}
\hypersetup{colorlinks=true,linkcolor=blue,citecolor=blue}
\usepackage{dsfont} 
\usepackage{nicefrac}

\title{\ed{Feasible Conditional Belief Distributions}
\footnote{We are grateful to Omer Tamuz for the discussions that have inspired this work.
The paper has benefited from our discussions with Emir Kamenica, Ce Liu, Alex Nesterov, Alex Smolin, and participants of ACM EC2022.}
}
\author{ \large Itai Arieli\thanks{Technion, Haifa (Israel). Itai Arieli has been supported by the Israel Science Foundation (grant \#2030524).}\ \ \ %
Yakov Babichenko\thanks{Technion, Haifa (Israel). Yakov Babichenko has been  supported by BSF grants \#2018397 and \#2021680.} \ \ \ %
Fedor Sandomirskiy\thanks{Princeton University, Princeton (USA). Fedor Sandomirskiy was  supported by the
Linde Institute at Caltech, PIMCO award, and the National Science Foundation (grant CNS 1518941). 	} 
}

\ifdefined\DRAFT

\definecolor{ForestGreen}{rgb}{.13,.54,.13}
\definecolor{BrickRed}{rgb}{.5,.01,.01}
\definecolor{violet}{cmyk}{0.79,0.88,0,0}
\newcommand{\fed}[1]{{\color{ForestGreen}{(\textbf{Fedor:} #1)}}}

\newcommand{\yakov}[1]{{\color{red}{(\textbf{Yakov:} #1)}}}
\newcommand{\ed}[1]{{{\color{BrickRed} {#1}}}}

\else
\newcommand{\fed}[1]{}
\newcommand{\yakov}[1]{}
\newcommand{\ed}[1]{#1}

\fi

\newtheorem{theorem}{Theorem}
\newtheorem{lemma}{Lemma}

\newtheorem{observation}{Observation}
\newtheorem{corollary}{Corollary}
\newtheorem{proposition}{Proposition}

\theoremstyle{definition}

\newtheorem{definition}{Definition}

\theoremstyle{remark}

\newtheorem{example}{Example}

\newcommand{\E}{\mathbb{E}}
\newcommand{\bP}{\mathbb{P}}

\newcommand{\R}{\mathbb{R}}
\newcommand{\val}{\mathrm{Val}}
\newcommand{\supp}{\mathrm{supp}\,}

\newcommand{\F}{\mathcal{F}}
\newcommand{\cav}{\mathrm{cav}}

\def\dd{\mathrm{d}}

\begin{document}

\ifdefined\EC
\begin{titlepage}
\maketitle
\end{titlepage}
\else
\maketitle

\begin{abstract}
\ed{Agents receive private signals about an unknown state. The resulting joint belief distributions are complex and lack a simple characterization.
Our key insight is that, when conditioned on the state, the structure of belief distributions simplifies: feasibility constrains only the marginal distributions of individual agents across states, with no joint constraints within a state.
We apply this insight to multi-receiver persuasion,  identifying new tractable cases and introducing optimal transportation and duality tools.
}

\vskip 0.3cm

\end{abstract}

\section{Introduction}

\ed{Bayesian persuasion is one of the major successes of information economics, which has reshaped how we model strategic information transmission. In the classical model of \cite*{kamenica2011bayesian}, a sender observes a realization of a random state and aims to influence a receiver's beliefs by selectively disclosing information through noisy signals. A vast body of literature aims to understand which insights of the basic model extend to more general settings.} 

\ed{Our paper examines the case of multiple receivers. In practical scenarios such as recommendation systems in electronic marketplaces or political campaigns targeting different voter groups, a sender interacts with multiple receivers individually, sending them private signals.
However, extending the persuasion model to multiple receivers presents significant challenges. In the single-receiver case, tractability arises from the simple structure of the set of feasible belief distributions---that is, the distributions of beliefs about the state that can be induced by some noisy signal. Optimal persuasion reduces to choosing an optimal feasible distribution, and the simplicity of this set leads to explicit solutions that can be obtained via the classical concavification technique of \cite*{aumann1995repeated}. In contrast, with multiple receivers, the set of feasible joint belief distributions has a complex structure \citep*{dawid1995coherent, arieli2020feasible,morris2020notrade,ziegler2020adversarial,lang2022feasible, cichomski2025existence}, which greatly limits tractability and makes classical techniques inapplicable.}
\medskip

\ed{Our paper augments the above negative message with a positive one.
We provide a new perspective on private multi-receiver persuasion by focusing on the characterization of feasible joint belief distributions \emph{conditional on the state}. Our central result (Theorem~\ref{th_feasibility}) shows that,
when conditioned on the realized state, the set of feasible belief distributions admits a simple characterization:
\begin{quote}
\emph{A collection of conditional joint belief distributions---one per each realized state---is feasible if and only if the corresponding one-receiver marginals are feasible.}
\end{quote}
In other words, feasibility only restricts one-receiver marginals across states and does not restrict correlation within a state. As a result,
checking feasibility in a multi-receiver problem reduces to checking feasibility in auxiliary single-receiver problems, one per agent. This contrasts sharply with the complexity of the unconditional joint belief distributions, which the literature has focused on.}
\medskip

\ed{We apply this insight to a class of multi-receiver persuasion problems, where the sender's objective depends on the induced receivers' beliefs and the state, as in the single-receiver model of \cite*{kamenica2011bayesian}. This framework is termed \emph{first-order persuasion} because higher-order beliefs do not affect receivers' actions. It serves as a benchmark for determining which insights extend to multi-receiver problems. 
However, even this seemingly simple setting proves largely intractable. 

\medskip
To tackle this intractability, we represent persuasion as the selection of optimal conditional belief distributions, leveraging their structure to formulate two linear programs---a primal and a dual. Each program highlights different aspects of the challenges in multi-receiver persuasion and suggests classes of problems where these challenges can be overcome.
\smallskip

The primal problem 
models persuasion as a two-stage optimization process. In the first stage, the sender selects a feasible belief distribution for each receiver, analogous to the single-receiver setting. In the second stage, the sender optimizes over all possible ways to correlate these marginal distributions. This optimization over correlations captures the essence and difficulty of multi-receiver persuasion. Finding the optimal correlation takes the form of a Monge–Kantorovich optimal transportation problem---a well-studied class of optimization problems for determining optimal joint distributions given marginals. A formal connection between persuasion and optimal transport was conjectured by \cite*{dworczak2019simple}.

The dual problem 
represents the sender's optimal value for a given utility function as the envelope of all utility functions that are pointwise above and where revealing no information is optimal. 
This result provides a multi-receiver generalization of the price-function duality approach developed by \cite*{dworczak2019persuasion}, now a standard tool in single-receiver persuasion. Our duality can be interpreted as a multi-receiver extension of the concavification formula by \cite*{kamenica2011bayesian} and is closely related to the Kantorovich duality in optimal transportation literature.
\medskip

We demonstrate how the primal and dual perspectives can be utilized to construct explicit solutions for new classes of first-order persuasion problems.

Using the primal approach, we show how to construct explicit solutions for problems where the sender's objective is non-trivial in only one state, as well as for problems with supermodular objectives. 
We discuss examples where a manager aims to support the morale of a group by ensuring that some agents remain optimistic even in bad states, 
a regulator selectively discourages production in a Cournot oligopoly with unknown costs, and a producer selectively discloses information about product quality to market segments to maximize an objective that depends on the adoption level in each segment.

The dual formulation offers a general guess-and-verify approach to first-order persuasion. We illustrate its application to polarization objectives and profit maximization by an informed retailer. In the appendix, we present a general methodology for ``guessing'' the solution to the dual problem guided by complementary slackness conditions. While we focus on examples with relatively simple optimal information structures, a follow-up paper by \cite*{kravchenko2024} demonstrates that a version of our duality approach also allows for solving problems where the optimal information structures are highly non-trivial.

\medskip
Prior to this work, explicit solutions to first-order persuasion problems were available only for specific settings---such as quadratic objectives (e.g., belief-covariance minimization), threshold objectives, or binary-actions \citep*{arieli2019private,ziegler2020adversarial,burdzy2020bounds,burdzy2021can,arieli2021feasible,cichomski2021maximal,cichomski2022contradictory,cichomski2022doob,smolin2022information,cichomski2023combinatorial}. Moreover, each of these tractable cases required an approach tailored to its particular context.  In contrast, our paper provides more flexible methods that are applicable beyond quadratic and threshold objectives.

\paragraph{Structure of the paper.} The paper is organized as follows: Section \ref{sec:feasible} characterizes feasible conditional belief distributions. Section \ref{sec_persuasion_as_transport} applies this characterization to persuasion. In Section \ref{sec:applications}, we explore several applications of our general results to particular classes of persuasion problems.  Finally, Section \ref{sec:conclude} concludes with directions for future research. Proofs and additional technical discussions are provided in the appendices.
}

\paragraph{Related literature.}
Optimal ways to persuade multiple receivers via private signals are known only for particular objectives and/or strong restrictions on receivers' action sets. The main obstacle is the complex structure of the set of feasible belief distributions (those joint distributions of beliefs that the sender can induce) as indicated by \cite*{dawid1995coherent,  mathevet2020information,arieli2020feasible,arieli2021feasible,he2021private,lang2022feasible}. Related feasibility questions were studied by \cite*{gutmann1991existence,herings2020belief, ziegler2020adversarial, levy2022persuasion,
morris2020notrade,brooks2019information,arieli2022population}.  Mathematical literature refers to feasible distributions as coherent distributions and provides some tight bounds, which can be converted into solutions to particular first-order persuasion problems \citep*{burdzy2020bounds,burdzy2021can,cichomski2020maximal,cichomski2021maximal, cichomski2022contradictory,cichomski2022doob,cichomski2023combinatorial}. 
 First-order persuasion simplifies dramatically if receivers have only a few actions; e.g., see \cite*{arieli2019private} for binary actions and sub/supermodular objectives.
 In general, for a few actions, one can identify signals with action recommendations satisfying incentive-compatibility constraints 
 and obtain the optimal information structure as a solution to a linear program capturing Bayesian correlated equilibria as in \cite*{bergemann2016bayes,bergemann2019information,taneva2019information}. 
 Our results are not sensitive to the cardinality of action sets and are applicable in the case of a continuum of actions.

The concavification technique of \cite*{kamenica2011bayesian} extends to multiple receivers if the sender is constrained to using public signals. Indeed, such problems reduce to persuading a single ``aggregate'' receiver as noted by, e.g., \cite*{laclau2017public}. \cite*{mathevet2020information} demonstrated the relevance of this insight even without any constraint on signals. They showed that a general persuasion problem can be decomposed into its ``private'' and  ``public'' components, and the public one can be tackled via concavification. 
Our multi-receiver concavification approach is of a different nature and captures both private and public components.
  
 A connection to optimal transportation is known in a variety of economic settings, e.g., monopoly pricing and  multi-dimensional screening \citep*{daskalakis2017strong,figalli2011multidimensional}, auctions \citep*{kolesnikov2022beckmann}, matching and labor market sorting~\citep*{chiappori2010hedonic, boerma2021sorting}, optimal taxation \citep*{steinerberger2019tax}, econometrics~\citep*{galichon2021survey}, and many others surveyed by
 \citep*{ekeland2010notes,Carlier, galichon2016optimal}.
 This connection is fruitful as it always brings new tools --- such as the Kantorovich duality --- from the mathematical theory of transportation to the problem of interest. The modern mathematical theory is surveyed by \citep*{BoKo, McCannGuill} and comprehensively presented in books \citep*{santambrogio2015optimal, villani2009optimal}.  
 
{ 
In parallel to our work\footnote{An extended abstract of this paper appeared in proceedings ACM EC2022 \citep*{arieli2022persuasion}. } connecting multi-receiver persuasion and transportation, several recent papers describe another connection for single-receiver problems  \citep*{kolotilin2022persuasion,cieslak2021optimal,malamud2021persuasion,lin2022credible}. 
In these papers, transportation problems arise as the optimal way to correlate the state and a recommendation to a single receiver, a perspective especially useful for continuous state spaces. By contrast, in our approach, the transportation problem captures the optimal correlation across the beliefs of multiple receivers, and we focus on finite sets of states.

The duality that we find in the multi-receiver setting can be seen as an extension of the general single-receiver duality by \cite*{dworczak2019persuasion}; see Section~\ref{sec_persuasion_as_transport} for a detailed comparison. Earlier duality results of \cite*{kolotilin2018optimal}, \cite*{dworczak2019simple}, and \cite*{dizdar2020simple} addressed the case of the sender's objective depending on the induced posterior mean.
The action-recommendation approach of \cite*{bergemann2016bayes} also leads to a linear program, and its dual is studied by \cite*{galperti2018dual} and \cite*{galperti2023value} for finite sets of actions. 
\cite*{smolin2022information} show that this dual problem gains tractability for a continuum of actions under extra convexity assumptions.
}

\section{\ed{Feasible Belief Distributions}}\label{sec:feasible}

\ed{
A random state $\omega$ is drawn from a finite set of states $\Omega$ according to a \emph{prior distribution} $p\in \Delta(\Omega)$ with full support. There are $n$ agents $N = \{1,2,\ldots,n\}$ who receive private signals about $\omega$ according to an information structure; we will often refer to them as receivers.

An \emph{information structure} $I=\big((S_i)_{i\in N},\pi(\,\cdot\mid \omega)\big)$ is composed of sets of signals $S_i$ 
for each receiver $i\in N$ and a joint distribution of signals $\pi(\cdot\mid \omega) \in \Delta(S_1\times \cdots\times S_n)$ conditional on each possible realization of the state $\omega$. The sets of signals can be arbitrary measurable spaces, i.e., sets equipped with sigma fields. 

Combined with the prior $p\in \Delta(\Omega)$, an information structure $I$ induces the joint distribution  $\bP=\bP_{I}$ of the state and signals $(\omega,\,s_1,\ldots,s_n)$. Each receiver $i$ is aware of the prior $p$ and the information structure $I$. 
Hence, having received her signal $s_i$, the receiver $i$ can compute her posterior belief $x_i\in\Delta(\Omega)$ about the state, i.e., $x_{i}(\omega)=\bP_{I}(\omega\mid s_i)$. The posterior belief is defined for almost all realizations of signals. 
For finite sets of signals, it can be computed  by the Bayes formula:
\begin{equation*}\label{eq_Bayes}
 x_{i}(\omega)=p(\omega)\cdot \frac{\pi(s_i\mid \omega)}{\sum_{\omega'\in \Omega} p(\omega')\cdot \pi(s_i\mid \omega')}.
 \end{equation*}
 Since the belief $x_i$ depends on a random signal $s_i$, it is a random variable itself with values in $\Delta(\Omega)$.
 Let $\mu_I$ be the joint distribution of $(x_1,\ldots, x_n)$ and $\mu^\omega_I$ be the joint distribution conditional on the state $\omega$, i.e.,
$$\mu^\omega_I(A)=\bP_I\Big((x_1,\ldots, x_n)\in A\mid \omega\Big)\quad \text{for any measurable $A\subset \Delta(\Omega)^N$},$$
and $\mu_I=\sum_\omega p(\omega)\cdot \mu^\omega_I$.
We focus on the question of what distributions can be obtained this way, i.e., induced by some $I$. We refer to such distributions as \emph{feasible}. Feasibility captures all belief distributions that are compatible with Bayesian updating under common prior.
\begin{definition}\label{def_feasible_conditional}
Given prior $p$, distributions  $(\mu^\omega)_{\omega\in \Omega}$ 
are \emph{feasible conditional distributions of beliefs} if $\mu^\omega=\mu^\omega_I$ for all $\omega \in \Omega$ and some information structure $I$. Similarly, $\mu$ is a \emph{feasible unconditional distribution of beliefs} if $\mu=\mu_I$ for some $I$.
\end{definition}

For the case of $n=1$ receiver,   conditional and unconditional distributions admit equally simple characterizations. As we will see, the simplicity of conditional distributions persists for $n\geq 2$ receivers, while the set of unconditional distributions becomes effectively intractable.

We first discuss the case of a single receiver. 
By the classical splitting lemma \citep*{aumann1995repeated,blackwell1951comparison,kamenica2011bayesian}, an unconditional distribution $\lambda\in \Delta(\Omega)$ of a single agent
is feasible if and only if it satisfies the martingale property: the average belief is equal to the prior $p$, i.e., $\int_{\Delta(\Omega)} x(\omega)\dd\lambda(x)=p(\omega)$ for all $\omega\in\Omega$. The set of all such distributions is denoted by $\Delta_p(\Delta(\Omega))$.

In the single-agent case, the conditional belief distributions are uniquely determined by unconditional ones \citep*[e.g.,][]{alonso2016bayesian,doval2024persuasion}. 
Let $\lambda_I\in \Delta(\Omega)$ be the unconditional belief induced by $I$. Then the conditional distribution of beliefs given $\omega$ is obtained by weighing $\lambda_I$ with the likelihood ratio of this state\footnote{This property follows from the fact that $P_I(\omega\mid x)=x(\omega)$---i.e., belief $x$ is the best prediction of the actual distribution given the available information---and the Bayes formula $P_I(x\mid \omega)=\frac{P_I(\omega\mid x)}{p(\omega)}\cdot P_I(x)$.}
\begin{equation}\label{eq_conditional_from_uncondiitonal}
\lambda_I^\omega(x)=\frac{x(\omega)}{p(\omega)}\cdot \lambda_I(x).
\end{equation}
Here and below, we write $\nu(x)=f(x)\cdot \tau(x)$ to indicate that a distribution $\nu$ is obtained from $\tau$ by weighing with a weight given by a function $f$. This identity can be understood literally for distributions with density or finite support; more generally, it means that the Radon-Nikodym derivative $\frac{\dd \nu}{\dd \tau}(x)=f(x)$.

Identity~\eqref{eq_conditional_from_uncondiitonal}, combined with the splitting lemma, results in a characterization of conditional feasibility for a single receiver.
\begin{observation}[conditional feasibility for a single receiver]\label{obs_one_receiver_feasibility}
Distributions
$(\lambda^\omega)_{\omega\in \Omega}$ over $\Delta(\Omega)$ are feasible conditional distributions for a single receiver if and only if 
\begin{equation}\label{eq_admissible}
\lambda^\omega(x)=\frac{x(\omega)}{p(\omega)}\cdot \lambda(x)\qquad \text{for some}\quad  \lambda\in \Delta_p(\Delta(\Omega)). 
\end{equation}
\end{observation}
This observation provides a convenient parametrization of feasible conditional distributions for one receiver by $\lambda\in \Delta_p(\Delta(\Omega))$. 
We note that knowing $\lambda^\omega$ for one $\omega$ pins down $\lambda$ and thus $\lambda^{\omega'}$ for all other states by~\eqref{eq_admissible}, resulting in an alternative criterion for feasibility: 
$\frac{x(\omega')}{p(\omega')}\cdot \lambda^\omega(x) =\frac{x(\omega)}{p(\omega)}\cdot \lambda^{\omega'}(x)$ for all $\omega,\omega'\in \Omega$.
\medskip

For $n\geq 2$ receivers, unconditional feasible distributions do not admit a simple characterization even for two receivers and a binary state \citep*{dawid1995coherent}. In addition to the martingale property, there is a new constraint coming from the impossibility of Bayesian-rational agents agreeing to disagree \citep*{aumann1976agreeing}. As a result $P_I\big(x_i(\omega)=1\  \text{and} \  x_j(\omega')=1\big)$ must be zero for $i\ne j$ and $\omega\ne \omega'$. There is, in fact, a continuum of hard-to-work-with constraints, indicating that even a partial disagreement cannot happen too often \citep*{arieli2020feasible}. 

Our main result shows that conditioning on the realized state disentangles the feasibility of individual belief distributions and the way these individual belief distributions are correlated. Consequently, the question of feasibility for conditional distributions retains simplicity: all the feasibility constraints originate from single-receiver problems. 

\begin{theorem}[conditional feasibility for $n$ receivers]\label{th_feasibility}
Distributions
$(\mu^\omega)_{\omega\in \Omega}$ are feasible conditional distributions if and only if the one-receiver marginals $(\mu_i^\omega)_{\omega\in \Omega}$ are feasible in a one-receiver problem for each receiver $i$.
\end{theorem}
We refer to this result as a theorem to highlight its importance despite its elementary proof.
\begin{proof}
One direction is immediate. If conditional distributions $(\mu^\omega)_{\omega\in \Omega}$ are feasible in an $n$-receiver problem, then each receiver's marginal distributions $(\mu_i^\omega)_{\omega\in \Omega}$ are necessarily feasible in a single-receiver problem.

For the opposite direction, we show that, conditional on the state, an information structure can correlate individual beliefs arbitrarily by correlating the corresponding signals.
Suppose that $(\mu^\omega)_{\omega\in \Omega}$ is such that $(\mu_i^\omega)_{\omega\in \Omega}$ are feasible in single-receiver problems for each~$i$. This means that, for each~receiver, there is an information structure $(S_i, \pi_i)$ inducing $(\mu_i^\omega)_{\omega\in \Omega}$. By the revelation principle, we can assume that signals are equal to induced beliefs, i.e., $S_i=\Delta(\Omega)$ and $x_i=s_i$. Now consider an $n$-agent information structure $I$ with a set of signals $T_i=\Delta(\Omega)$ for receiver $i$ and the joint distribution of signals $\pi(\,\cdot \mid \omega)=\mu^\omega$. Let $t_1,\ldots, t_n$ be the realized signals and $y_1,\ldots, y_n$ be the induced beliefs. By the construction, the  distributions of $(s_i,\omega)$ and $(t_i,\omega)$ are the same. Consequently, the belief induced by the signal $t_i$ equals the signal, i.e., $y_i=t_i$. Since the conditional joint distribution of signals is $\mu^\omega$ and signals coincide with beliefs,
we conclude that the conditional belief distribution $\mu^\omega_I=\mu^\omega$ and thus $(\mu^\omega)_{\omega\in \Omega}$ are feasible.
\end{proof}

To illustrate why the agreeing-to-disagree constraint disappears for conditional distributions---thus reconciling their simplicity with the complexity of unconditional distributions---we present the following example in the binary-state case.
\begin{example}[conditional vs.~unconditional feasibility]\label{ex_conditional_vs_unconditional}
Consider a binary state $\omega\in\Omega=\{\ell,h\}$ and 
represent each belief $x\in \Delta(\Omega)$ by the weight assigned to state~$\ell$, i.e., $x(\ell)\in[0,1]$. Assume that the prior is \(\frac{1}{2}\), and the two agents receive symmetric binary signals that match the state with probability \(r > \frac{1}{2}\). Thus, the possible pairs $(x_1, x_2)$ of induced beliefs are \((r, r),\ (r,\, 1 - r), \ (1 - r,\, r), \text{ and } (1 - r,\, 1 - r)\). As the correlation between signals is not fixed---e.g., signals can be identical or conditionally independent---there is a range of joint belief distributions they can induce. 

For simplicity, we focus on symmetric distributions, where the weights of the off-diagonal beliefs $(r,\,1-r)$ and $(1-r,\, r)$ are equal to $\alpha\geq 0$. Therefore, unconditional belief distributions $\mu\in \Delta([0,1]^2)$ satisfying the martingale constraint are of the form
$$\mu= {(1-\alpha)}\left(\delta_{(r, r)}+\delta_{(1-r, 1-r)}\right) + {\alpha} \left(\delta_{(r, 1-r)}+\delta_{(1-r, r)}\right),$$
where~$\delta_z$ denotes a point mass at a point~$z$. Such $\mu$ is feasible if the weight on each of the ``disagreement outcomes'' $\alpha\leq 1-r$ \citep*[][Proposition~1]{arieli2020feasible}. In particular, the correlation between posteriors becomes almost perfect as signals precision $r$ approaches~$1$. This phenomenon is a repercussion of the general constraints on disagreement for unconditional feasible distributions. 

\begin{figure}[h!]
    \begin{center}
        \scalebox{0.6}{
            \begin{tikzpicture}[scale=0.5, line width=1.5pt]
                \draw (0,0)--(10,0)--(10,10)--(0,10)--(0,0);
                \draw[thin, densely dashed] (2,0)--(2,10);
                \draw[thin, densely dashed] (8,0)--(8,10);
                \draw[thin, densely dashed] (0,2)--(10,2);
                \draw[thin, densely dashed] (0,8)--(10,8);
                \node[below right] at (10,0) {\Large $x_1(\ell)$};
                \node[above left] at (0,10) {\Large $x_2(\ell)$};
                \node[below] at (2,0) {\Large $1-r$};
                \node[below] at (8,0) {\Large $r$};
                \node[left] at (0,2) {\Large $1-r$};
                \node[left] at (0,8) {\Large $r$};
                \filldraw[red] (8,8) circle (0.9); 
                \filldraw[red] (8,2) circle (0.4); 
                \filldraw[red] (2,8) circle (0.4); 
                \node at (5,12.5) {\huge $\mu^\ell:$};
            \end{tikzpicture}
            \hspace{1cm}
            \begin{tikzpicture}[scale=0.5, line width=1.5pt]
                \draw (0,0)--(10,0)--(10,10)--(0,10)--(0,0);
                \draw[thin, densely dashed] (2,0)--(2,10);
                \draw[thin, densely dashed] (8,0)--(8,10);
                \draw[thin, densely dashed] (0,2)--(10,2);
                \draw[thin, densely dashed] (0,8)--(10,8);
                \node[below right] at (10,0) {\Large $x_1(\ell)$};
                \node[above left] at (0,10) {\Large $x_2(\ell)$};
                \node[below] at (2,0) {\Large $1-r$};
                \node[below] at (8,0) {\Large $r$};
                \node[left] at (0,2) {\Large $1-r$};
                \node[left] at (0,8) {\Large $r$};
                \filldraw[blue] (8,2) circle (0.4); 
                \filldraw[blue] (2,8) circle (0.4); 
                \filldraw[blue] (2,2) circle (0.9); 
                \node at (5,12.5) {\huge $\mu^h:$};
            \end{tikzpicture}
            \hspace{1cm}
            \begin{tikzpicture}[scale=0.5, line width=1.5pt]
                \draw (0,0)--(10,0)--(10,10)--(0,10)--(0,0);
                \draw[thin, densely dashed] (2,0)--(2,10);
                \draw[thin, densely dashed] (8,0)--(8,10);
                \draw[thin, densely dashed] (0,2)--(10,2);
                \draw[thin, densely dashed] (0,8)--(10,8);
                \node[below right] at (10,0) {\Large $x_1(\ell)$};
                \node[above left] at (0,10) {\Large $x_2(\ell)$};
                \node[below] at (2,0) {\Large $1-r$};
                \node[below] at (8,0) {\Large $r$};
                \node[left] at (0,2) {\Large $1-r$};
                \node[left] at (0,8) {\Large $r$};
               \draw[very thin,fill=red, shift={(8,2)}] (0,0) -- (-45:0.4) arc(-45:135:0.4) -- cycle;
                \draw[very thin,fill=blue, shift={(8,2)}] (0,0) -- (135:0.4) arc(135:315:0.4) -- cycle;
                \draw[very thin,fill=red, shift={(2,8)}] (0,0) -- (-45:0.4) arc(-45:135:0.4) -- cycle;
                \draw[very thin,fill=blue, shift={(2,8)}] (0,0) -- (135:0.4) arc(135:315:0.4) -- cycle;
                %
                \filldraw[blue] (2,2) circle (0.7); 
                \filldraw[red] (8,8) circle (0.7); 
               \node at (5,12.5) {\huge $\mu=\frac{1}{2}\mu^\ell+ \frac{1}{2}\mu^h:$};
            \end{tikzpicture}
        }
    \end{center}
    \caption{\ed{Conditional and unconditional belief distributions for Example~\ref{ex_conditional_vs_unconditional} placing as much weight on disagreement outcomes as permitted by feasibility. For accuracy $r$ close to $1$, the marginals force $\mu^\ell$ to put almost all weight on $(r,r)$ and $\mu^h$, on $(1-r,1-r)$. As a result, beliefs become almost perfectly correlated under $\mu$.      
    Red/blue colors correspond to $\omega=\ell$ and $\omega=h$.}
    \label{fig_cond_vs_uncond}}
\end{figure}

We now consider the conditional distributions $\mu^\ell$ and $\mu^h$ on $[0,1]^2$. By Theorem~\ref{th_feasibility}, they are feasible if
and only if the marginal probabilities of posteriors $1-r$ and $r$ are 
$\mu_i^\ell(\{r\})=r$, $\mu_i^\ell(\{1-r\})=1-r$ and $\mu_i^h(\{r\})=1-r$, $\mu_i^h(\{1-r\})=r$.  Notably, no constraints on the joint distribution within each state---such as constraints on the correlation of beliefs---are needed.

The constraint $\alpha\leq 1-r$ for unconditional feasibility follows immediately from the constraints on marginals of $\mu^\ell$ and $\mu^h$. Indeed, the weights put by $\mu^\ell$ on each of the disagreement outcomes $(1-r,\, r)$ and $(r,\, 1-r)$ cannot exceed the one-agent marginal probability of the belief $1-r$, and thus does not exceed $1-r$. 
Similarly, for $\mu^h$, the weights on disagreement outcomes cannot exceed $1-r$. As a result, the total weight $\alpha$ placed by the unconditional distribution $\mu=\frac{1}{2}\mu^\ell+\frac{1}{2}\mu^h$ on disagreement outcomes cannot exceed $1-r$ as well. See Figure~\ref{fig_cond_vs_uncond} illustrating the construction of $\mu$ placing as much weight on disagreement outcomes as permitted by feasibility.

For accuracy $r$ close to $1$, the marginals of $\mu^\ell$ put almost all weight on $x_i=r$, and thus $\mu^\ell$ is concentrated on $(r,r)$, placing little weight on other combinations of posteriors.
Similarly, $\mu^h$ is concentrated on $(1-r,1-r)$. As a result, the unconditional distribution $\mu=\frac{1}{2}\mu^\ell+\frac{1}{2}\mu^h$ places most of the weight on the diagonal, i.e., the constraint on the correlation in the unconditional distribution originates from averaging distributions that tend to concentrate at a
single diagonal point due to the concentration of marginals.

We conclude that the joint constraint on receivers' unconditional belief distribution originates from much simpler individual constraints on single-receiver marginals conditional on the state.

\end{example}

}

\section{Implications for Persuasion}\label{sec_persuasion_as_transport}

A \emph{first-order Bayesian persuasion problem} is specified by the collection
$$B=\Big(\Omega,\ p\in\Delta(\Omega),\ N,\  v\,:\, \Omega\times \big(\Delta(\Omega)\big)^N\to\R\Big),$$
where $\Omega$ is the set of states, $p$ is the prior distribution,
$N$ is the set of receivers, and $v$ is the sender's utility function, which depends on the state and the receivers' beliefs. \ed{We assume that~$v^\omega$ is upper semicontinuous in beliefs $(x_1,\ldots,x_n)$ for each state $\omega$.}
The sender observes the realized state $\omega$ and can selectively reveal information about $\omega$ to the receivers, who do not observe the realization of $\omega$ but are aware of the prior.
The sender's goal is to maximize the expected utility
$$ \mbox{maximize}\quad \E_I[v^\omega(x_1,\ldots,x_n)]$$ over all information structures $I$. 
\smallskip

\ed{
In first-order persuasion, the sender's objective does not depend on the receivers' higher-order beliefs.
Such utility functions arise as indirect utilities if each receiver~$i$ has an action set $A_i$ and the receiver's belief $x_i$ is a sufficient statistic for her action $a_i=a_i(x_i)$. For example, this is the case if there are no strategic externalities across receivers, i.e., each receiver's utility depends solely on their own action and the state.  We will discuss examples where the sender is a manager aiming to sustain workers' morale by ensuring some remain optimistic even in a bad state  
or a producer selectively disclosing information about product quality to different market segments to maximize an objective that depends on the adoption level in each submarket. 

First-order persuasion also arises in settings with externalities when considering the bounded rationality of receivers. Receivers may be agnostic about opponents' beliefs, as in \cite*{ziegler2020adversarial}, or they may update on their own signals without anticipating that opponents' actions also reflect updated information, as in the cursed equilibrium of \cite*{eyster2005cursed}. We will explore such examples in the contexts of selective production discouragement in a Cournot oligopoly with unknown costs
and informed retailer profit maximization. 
}

\ed{
\medskip
The optimal value of the sender's objective is called the value of the persuasion problem~$B$: 
\begin{equation}\label{eq_persuasion_value}
 \val[B]=\max_{I} \,\E_{I}\Big[v^\omega(x_1,x_2\ldots,x_n)\Big].
 \end{equation}
We write $\max$ instead of $\sup$ as the existence of an optimal information structure $I$ is guaranteed thanks to the upper semicontinuity of $v$; see Appendix~\ref{app_existence_primal}. 

We derive primal and dual representations for the sender's optimal value as corollaries of Theorem~\ref{th_feasibility}. In Section~\ref{sec:applications}, we will show how these representations can be used to find closed-form solutions in various examples.

Maximization over information structures
 is equivalent to maximizing over joint distributions of the state $\omega$ and posterior beliefs $x_1,\ldots,x_n$ that can be induced by some information structure~$I$, i.e., over feasible conditional belief distributions. 
 We conclude that the value of the persuasion problem admits the following representation
\begin{equation}\label{eq_value_as_sup_over_feasible}
\val[B]=\max_{{\footnotesize \begin{array}{c} \mbox{{feasible}}\\ 
(\mu^\omega)_{\omega\in \Omega}
\end{array}}} \left( \sum_{\omega\in \Omega} \ p(\omega)\cdot \int_{\Delta(\Omega)\times\ldots \times\Delta(\Omega)} v^\omega\,\dd\mu^\omega\right).
\end{equation}
Combining this representation with the characterization of feasible conditional distributions established in Theorem~\ref{th_feasibility}, we obtain the following corollary.
\begin{corollary}[primal value representation]\label{cor_persuasion_as_transport}
The value of a persuasion problem can be expressed as:
\begin{equation}\label{eq_persuasion_transport}
\val[B]=\max_{{\footnotesize \begin{array}{c} \mbox{\emph{1-agent feasible}}\\ 
(\lambda_{i}^\omega)_{\omega\in \Omega}, \ i\in N
\end{array}}} \left(\sum_{\omega\in \Omega} \ \ p(\omega)\cdot \max_{{\footnotesize \begin{array}{c} 
\mbox{\emph{all distributions}}\\ 
\mbox{\emph{$\pi$ {with marginals}}}\\ 
\pi_i=\lambda_{i}^\omega, \ i\in N
\end{array}}} \int_{\Delta(\Omega)\times\ldots \times\Delta(\Omega)} v^\omega\,\dd\pi\right).
\end{equation}
\end{corollary}	
This formula represents multi-receiver persuasion as a two-stage optimization. In the first stage, the sender selects individual belief distributions $\lambda_i^\omega$ for each receiver $i$, as in a single-receiver problem. In the second stage, the sender finds the optimal way to correlate these individual distributions by choosing a joint distribution with marginals $\lambda_i^\omega$. Conditioning on the state $\omega$ ensures that the choices made in the first stage do not restrict possible correlations in the second, allowing maximization over all joint distributions consistent with the marginals.
\smallskip

The problem of finding the optimal way to correlate given marginal distributions is known as the Monge-Kantorovich optimal transportation problem. It is given by a measurable utility function $v$ on $X_1\times\ldots \times X_n$ and a collection of probability measure $\lambda_i\in \Delta(X_i)$ for each $i\in N=\{1,\ldots, n\}$. The goal is to find the joint distribution $\pi$ with marginals $\lambda_i$ that maximizes the integral of $v$:
 $$ 
 MK_v\big[(\lambda_i)_{i\in N}\big ]=\max_{{\footnotesize \begin{array}{c} \pi\in \Delta(X_1\times\ldots \times X_n)\\ 
\mbox{with marginals $\pi_i=\lambda_i$}
\end{array}}} \int_{X_1\times\ldots \times X_n} v\,\dd\pi
 $$
The term transportation originates from the two-marginal interpretation, where $\lambda_1\in \Delta(X_1)$ represents a spacial distribution of production of a certain commodity, $\lambda_2\in \Delta(X_2)$ captures the distribution of consumption, $c=-v$ captures the cost of transporting a unit amount of the commodity from one location to the other, and the goal is to find the least costly transportation plan~$\pi$ such that supply meets demand.

\smallskip 
We conclude that the internal maximization in~\eqref{eq_persuasion_transport} is a  transportation problem $MK_{v^\omega}\big[(\lambda_i^\omega)_{i\in N}\big ]$
with $X_1=\ldots=X_n=\Delta(\Omega)$, utility~$v^\omega$, and marginals $\lambda_1^\omega,\ldots, \lambda_n^\omega$. In other words, a persuasion problem is equivalent to a family of transportation problems with a joint constraint on marginals.} 
\smallskip

An essential tool in optimal transportation theory is the dual representation of the optimal value, known as the Kantorovich duality. Drawing inspiration from this classical result, we derive a dual representation for the sender's optimal value. This new representation not only generalizes the single-receiver duality established by \cite*{dworczak2019persuasion} but also extends the celebrated concavification formula to the multi-receiver context.
\begin{proposition}[dual value representation]\label{prop_dual value representation}
	The value of a persuasion problem can be expressed as:
	\begin{align}\label{eq:dual_value_explicit}
			\val[B]=&\inf_{{\footnotesize \begin{array}{c}  \mbox{\emph{$V^\omega\in\R$, continuous  $\varphi_i^\omega$ on $\Delta(\Omega)$ such that}}\\
			\mbox{$v^\omega(x_1,\ldots,x_n)\leq V^\omega+\sum_{i\in N}\varphi_i^\omega(x_i)$}\\
					\mbox{\emph{and} $\sum_{\omega\in\Omega}x_i(\omega) \varphi_i^\omega(x_i)=0$} 	\end{array}}}\sum_{\omega\in\Omega} p(\omega)\cdot V^\omega.
		\end{align}
If $v^\omega$ is continuous,  the optimum is attained, i.e., $\inf$ can be replaced with $\min$.
\end{proposition}
The proposition is proved in  Appendix~\ref{app_proof_of_duality}, where we also show that functions $\varphi_i^\omega$ inherit the regularity of $v$, e.g., if $v$ is Lipshitz, so are the dual variables. We prove the duality by starting with the persuasion problem expressed as maximization over feasible conditional distributions and derive its dual along the lines of the duality proof in the optimal transport literature.  We define an auxiliary zero-sum game with a sup-inf value equal to the value of the persuasion problem, use Sion's minimax theorem to exchange sup and inf, and show that the inf-sup value coincides with the right-hand side of \eqref{eq:dual_value_explicit}.

\paragraph{Comparison to the single-receiver duality.}
Consider a persuasion problem with one receiver and the utility function $v^\omega=v$ independent of the state.  
\cite*{dworczak2019persuasion} established\footnote{For  a finite dimension $(|\Omega|<\infty)$, the result is intuitive. The value is known to be equal to the concavification $\cav[v](p)$, and the concavification of a function is the envelope of affine functions that lie above it. \cite*{dworczak2019persuasion} demonstrated that this remains true in the far less intuitive infinite-dimensional case, e.g., for continuous $\Omega$.} a dual representation for the value, which, in our notation, can be written as follows:
\begin{equation}\label{eq_Dworzak_Kolotilin_duality}
    		\val[B]=\inf_{{\footnotesize \begin{array}{c}  \mbox{$V^\omega\in\R$ such that}\\
			\mbox{$v(x)\leq \sum_{\omega\in\Omega}x(\omega)\cdot V^\omega$}
				\end{array}}} \sum_{\omega\in\Omega} p(\omega)\cdot V^\omega.
\end{equation}
 The crucial difference between \eqref{eq_Dworzak_Kolotilin_duality} and Proposition~\ref{prop_dual value representation} is that functions $\varphi_i^\omega$ are absent in the single-receiver case.
Consequently, the problem with one receiver is finite-dimensional, while that from Proposition~\ref{prop_dual value representation} is infinite-dimensional.

One may wonder if we can assume that $\varphi_i^\omega\equiv 0$  in Proposition~\ref{prop_dual value representation}.
The answer is negative for more than one receiver, even if the utility function is state-independent and satisfies all the symmetries. In Section~\ref{sect_polarization}, we will see an example with two receivers, where the optimum is attained at non-linear functions $\varphi_i^\omega$. We believe that, as in the theory of optimal transportation, the minimization cannot be restricted to functions $\varphi_i^\omega$ having a simple parametric form. This can be seen as another justification for the difficulty of multi-receiver persuasion.

\paragraph{{Proposition~\ref{prop_dual value representation} as an analog of the concavification formula.}}
Consider a single-receiver persuasion problem $B$ with a continuous state-independent utility $v$. 
The value of this problem is equal to the concavification $\cav[v](p)$ \citep*{kamenica2011bayesian}. Notice that $u=\cav[v]$ is a concave continuous function and, in particular, revealing no information would be optimal if the utility function were equal to $u$ instead of $v$. Hence, the classical concavification result can be restated as follows. For a single receiver, 
\begin{align}\label{eq_single_agent}
\val[B]=\min_{{\footnotesize \begin{array}{c}  \mbox{{continuous $u$ such that}}\\
		\mbox{$v\leq u$ {and}}\\
		\mbox{{non-revealing is optimal under $u$ for any prior}} 	\end{array}}} 
		u(p).
\end{align}
Moreover, one can restrict minimization to linear $u$. 

\ed{In this form, the result remains valid for any number of receivers and state-dependent utilities. For an $n$-receiver persuasion problem $B$ 
with continuous and possibly state-dependent utility $v$, and prior~$p$,
\begin{equation}\label{eq_analog_cavU}
\val[B]=\min_{{\footnotesize \begin{array}{c}  \mbox{{continuous $u$ such that}}\\
		\mbox{$v^\omega(x_1,\ldots,x_n)\leq u^\omega(x_1,\ldots,x_n)$ {and}}\\
		\mbox{{non-revealing is optimal under  $u$ for any prior}} 	\end{array}}} 
		\sum_{\omega\in \Omega} p(\omega)\cdot u^\omega(p,p,\ldots,p).
\end{equation}
This formula provides an alternative form of
Proposition~\ref{prop_dual value representation}. It is straightforward to see that the right-hand side of~\eqref{eq_analog_cavU} is an upper bound on the left-hand side. For the reverse inequality, consider utilities $u^\omega$ of the form: \begin{equation}\label{eq_u} u^\omega(x_1,\ldots, x_n)=V^\omega+\sum_{i \in N} \varphi_i^\omega(x_i)\quad\text{with}\quad \sum_{\omega\in \Omega} x_i(\omega) \varphi_i^\omega(x_i)=0 \text{ for all } i\in N. \end{equation} By Proposition~\ref{prop_dual value representation}, the value of a persuasion problem with such a utility function $u$ under any prior~$q$ cannot exceed $\sum_{\omega \in \Omega} q(\omega) \cdot V^\omega$. The constraints imposed on $\varphi_i^\omega$ allow us to rewrite this sum as $\sum_{\omega \in \Omega} q(\omega)\cdot u^\omega(q,\ldots, q)$. This upper bound is achieved by non-revealing, which is therefore optimal under any prior $q$. Taking $V^\omega$ and $\varphi_i^\omega$ to be the optimal variables from the dual value representation for the original problem $B$ implies the equality in~\eqref{eq_analog_cavU}.

As a result, minimization in~\eqref{eq_analog_cavU} can be restricted to separable utilities $u^\omega$ of the form~\eqref{eq_u}, which serve as multi-receiver analogs of linear objectives in the single-receiver case of~\eqref{eq_single_agent}. It is important to note, however, that in the multi-receiver case, the condition that non-revealing is optimal cannot simply be replaced by concavity: taking $u = \cav[v]$ in~\eqref{eq_analog_cavU} generally leads to a strict upper bound on $\val[B]$. For an extreme but illustrative example, consider a binary-state problem with prior $p = {1}/{2}$ and a utility $v$ equal to one only at the pairs $(0,1)$ and $(1,0)$, and zero elsewhere. The concavification $\cav[v]$ evaluated at $(p, p)$ is one, but the actual value $\val[B]$ is zero because assigning positive mass to $(0,1)$ and $(1,0)$ violates Aumann's impossibility of agreeing to disagree. Smoothing $v$ around the discontinuities would not alter the conclusion.}

 \paragraph{Comparison to the Kantorovich duality.} Kantorovich found the dual to the transportation problem in the case of two marginals. The multi-marginal version of the result is as follows:\footnote{\ed{In the transportation literature, the dual is commonly expressed in an equivalent form that incorporates constants \( V^\omega \) within the functions \( \varphi_i \) \citep*{rachev1998mass}. We single out \( V^\omega \) to highlight the resemblance to the persuasion dual.
}}
	\begin{equation}\label{eq_Kantorovich}
	MK_v\big[(\lambda_i)_{i\in N}\big ]=\inf_{{\footnotesize \begin{array}{c}  \mbox{$V\in\R,$ continuous $\varphi_i:\, X_i\to\R$}\\
			\mbox{such that $v(x_1,\ldots,x_n)\leq V+\sum_{i\in N}\varphi_i(x_i)$}\\
\mbox{and $\int_{X_i} \varphi_i(x_i)\dd\lambda_i(x_i)=0$}			
			\end{array}}}
			V,
	\end{equation}
	where $X_i$, $i\in N$, are compact metric spaces and $v$ is an upper semicontinuous function on their Cartesian product; the optimum exist provided that $v$ is continuous \citep*{rachev1998mass}.

 The similarity between Proposition~\ref{prop_dual value representation} and the Kantorovich duality is not surprising, thanks to the connection between primal persuasion and transportation problems (Corollary~\ref{cor_persuasion_as_transport}). The differences are caused by the fact that the marginals in Corollary~\ref{cor_persuasion_as_transport}are not fixed but are free parameters that satisfy the one-receiver feasibility constraints. Hence, in contrast to \eqref{eq_Kantorovich}, the marginals do not enter~\eqref{eq:dual_value_explicit} and the functions $\varphi_i^\omega$ are required to satisfy the pointwise orthogonality  requirement $\sum_{\omega\in\Omega}x_i(\omega) \varphi_i^\omega(x_i)=0$ instead of functional orthogonality to measures $\lambda_i$ as in \eqref{eq_Kantorovich}.

\section{Applications}\label{sec:applications}

We rely on the primal and dual approaches to multi-receiver persuasion discussed in the previous section (Corollary \ref{cor_persuasion_as_transport} and Proposition \ref{prop_dual value representation})  to construct explicit solutions to several new classes of persuasion problems.

\subsection{One-state persuasion}\label{sect_onestate} 
A problem $B$ is a \emph{one-state persuasion problem} if the sender's   utility function $v^\omega$ has the following form 
$$v^\omega(x_1,\ldots,x_n)=\left\{\begin{array}{cc}v(x_1,\ldots,x_n), & \omega=\omega_0\\ 0, & \omega\ne\omega_0\end{array}\right.,$$
where $\omega_0\in \Omega$ is fixed and $v$ is a function $\Delta(\Omega)^N\to \R$.
For one-state problems, only the state $\omega_0$ contributes to the formula for the value from Corollary~\ref{cor_persuasion_as_transport}. \ed{Accordingly, only the component of one-agent marginals corresponding to this state plays a role.
We say that a distribution $\lambda^{\omega_0}$  on~$\Delta(\Omega)$ is a \emph{feasible one-agent marginal at a state $\omega_0$} if there is a way to complete it to a collection of feasible one-agent marginals  $(\lambda^\omega)_{\omega\in \Omega}$. 
\begin{lemma}\label{lm_projection_admissible} A distribution $\lambda^{\omega_0}$ is a feasible one-agent marginal at a state $\omega_0$
if and only if 
\begin{align}\label{eq_admissible_marginals_omega0}
		\int_{\Delta(\Omega)} \frac{x(\omega)}{x(\omega_0)}\dd\lambda^{\omega_0}(x)   &\leq \frac{p(\omega)}{p(\omega_0)}\quad \text{for all}\quad \omega\in\Omega\setminus\{\omega_0\}. 
	\end{align} 
\end{lemma}
The lemma is proved in Appendix~\ref{app_one_state}. The necessity of conditions~\eqref{eq_admissible_marginals_omega0} is easy to see. By Observation~\ref{obs_one_receiver_feasibility}, the feasibility of 
$(\lambda^\omega)_{\omega\in \Omega}$ implies that 
$\frac{x(\omega)}{x(\omega_0)} \lambda^{\omega_0}(x)=\frac{p(\omega)}{p(\omega_0)} \lambda^{\omega}(x)$. Integrating this identity over  $x\in \Delta(\Omega)$ with $x(\omega_0)\ne 0$, we get the equality of the form~\eqref{eq_admissible_marginals_omega0} with an extra multiplicative factor $\lambda^\omega(\{x\colon x(\omega_0)\ne 0\})$ on the right-hand side. By dropping this factor, we obtain the inequality.}

\smallskip
By combining Corollary~\ref{cor_persuasion_as_transport} and Lemma~\ref{lm_projection_admissible}, we conclude that the value of a one-state persuasion problem can be represented as
\[
p(\omega_0) \cdot \max_{\pi} \int_{\Delta(\Omega)\times\ldots\times \Delta(\Omega)} v \, \dd\pi,
\]
where the maximization is over distributions \(\pi \in \Delta\big(\Delta(\Omega) \times \cdots \times \Delta(\Omega)\big)\) whose marginals \((\pi_i)_{i \in N}\) satisfy the inequalities~\eqref{eq_admissible_marginals_omega0}. Since we maximize a linear functional over a convex set, we can, by Bauer's principle, restrict attention to the extreme points of this set. These extreme points have a simple structure: they are convex combinations of at most \( |N| \cdot (|\Omega| - 1) + 1 \) point masses, because the feasible set is defined by \( |N| \cdot (|\Omega| - 1) \) linear inequalities intersecting the simplex of all probability measures. The following lemma formalizes this observation.

\begin{lemma}\label{lm_one_state_counting_extreme}
	The value of a one-state persuasion problem $B$ can be expressed as maximization over distributions $\pi$ supported on at most $|N|\cdot(|\Omega|-1)+1$ points:
	\begin{equation}\label{eq_value_as_finite_dim_optimization}
	\val[B]=p(\omega_0)\cdot \max_{{\footnotesize \begin{array}{c} \pi\in \Delta\big(\Delta(\Omega)\times\ldots\times \Delta(\Omega)\big)\\ \mbox{\emph{such that the marginals satisfy} \eqref{eq_admissible_marginals_omega0} \emph{and}}\\
			\big|\supp[\pi]\big|\leq |N|\cdot(|\Omega|-1)+1
			\end{array}}	
	}\int_{\Delta(\Omega)\times\ldots\times \Delta(\Omega)} v\,\dd\pi.
	\end{equation}
\end{lemma}

Note that for $\pi$ from the lemma, the integral in \eqref{eq_value_as_finite_dim_optimization} as well as the integrals in \eqref{eq_admissible_marginals_omega0} are, in fact, finite sums with at most $|N|(|\Omega|-1)+1$ summands. In Appendix~\ref{app_one_state}, we prove a strengthening of Lemma~\ref{lm_one_state_counting_extreme} with the bound on the number of atoms depending on the number of ``active'' constraints~\eqref{eq_admissible_marginals_omega0}; we also demonstrate that the sender can achieve 
the utility level corresponding to a distribution $\pi$ by using an information structure with at most $|N|\cdot(|\Omega|-1)$ signals per receiver.

The possibility of reducing one-state persuasion to a finite-dimensional problem reflects a peculiar geometric property of the set of feasible conditional belief distributions $(\mu_\omega)_{\omega\in\Omega}$. Denote this set by $\F$. The set of distributions with marginals satisfying~\eqref{eq_admissible_marginals_omega0} can be seen as the projection of $\F$ under the map $(\mu^\omega)_{\omega\in \Omega}\to \mu^{\omega_0}$. The fact that this image has extreme points with finite support and a simple structure is to be contrasted with the complicated structure of extreme points of the set $\F$ itself.
Indeed, extreme feasible unconditional distributions (i.e., the image of $\F$ under  $(\mu^\omega)_{\omega\in \Omega} \to \sum_{\omega\in \Omega} p(\omega)\cdot \mu^\omega$) can have infinite support \citep*{arieli2020feasible, zhu2022some} and even be non-atomic \citep*{cichomski2025existence}. Since an extreme point of the image under a linear map is the image of an extreme point, this implies the existence of infinitely-supported extreme points in $\F$ itself.

\medskip
We now illustrate the application of Lemma~\ref{lm_one_state_counting_extreme} to particular persuasion problems.

\begin{example}[\ed{supporting group morale in bad states}]\label{ex_one_state}
\ed{Consider a principal who wants to support the morale of a group of two agents by revealing information about a binary state, which can be good ($\omega=h$) or bad ($\omega=\ell$). In the bad state, the principal {needs at least one of the agents to believe that the state is likely to be good} to maintain group morale. For example, the state can indicate how promising a project is in its current condition. {Workers' effort is increasing in their belief about the good state}, and it is particularly important to ensure that at least one worker invests high effort in the bad state.}
\begin{figure}[h!]
	\begin{center}
	    \scalebox{0.6}{
	    
		\begin{tikzpicture}[scale=0.5,line width=1.5pt]
		\draw (0,0)--(10,0)--(10,10)--(0,10)--(0,0);

		\draw[thin, densely dashed] (3.33333,0)--(3.33333,10);
  \draw[thin, densely dashed] (0,3.33333)--(10,3.33333);
		
		
		
		\filldraw[blue] (3.333333,3.333333) circle (0.6);
		\filldraw[red] (3.333333,10) circle (0.4);
  		\filldraw[red] (10,3.333333) circle (0.4);
            \node[right] at (3.9, 8.7) {\huge$\frac{1}{4}$};
		\node[left] at (9.8, 4.5) {\huge$\frac{1}{4}$};
		\node[right] at (3.9,4.5) {\huge$\frac{1}{2}$};
		\node[below right] at (10,0) {\Large $x_1(\ell)$};
		\node[above left] at (0,10) {\Large $x_2(\ell)$};
		\node[below] at (3.333333,0) {\huge $\frac{1}{3}$};
		\node[left] at (0,3.333333) {\huge $\frac{1}{3}$};
		\end{tikzpicture}
        }
	\end{center}
	\caption{The joint distribution of beliefs for Example~\ref{ex_one_state}. Prior is $1/2$. The numbers inside the square indicate the probabilities of each outcome, and red/blue colors correspond to $\omega=\ell$ and $\omega=h$, respectively.\label{fig_one_state}}
\end{figure}
\ed{For simplicity, we assume that the two states $\omega \in \{\ell,h\}$ are equally likely. We denote the belief of agent $i$ about the low state by $x_i \in [0,1]$. Agent $i$ exerts effort proportional to her belief $1-x_i$ about the good state. Principal's utility is proportional to the maximal effort in the bad state and is constant in the high state: $v^\ell (x_1,x_2)=\max\{1-x_1,1-x_2\}=1-\min\{x_1,x_2\}$  and  $v^h (x_1,x_2)=C$ with some $C$ in the good state. As the constant does not affect the optimum, we can set $C=0$, to obtain a one-state persuasion problem.
}

For two receivers, it is enough to consider distributions $\pi$  in~\eqref{eq_value_as_finite_dim_optimization} with at most three points in the support. 
If we restrict the maximization to one-point distributions, then the optimum of $1-\frac{1}{\sqrt{2}}\approx 0.29$ is achieved at the point mass at a pair of beliefs $(x_1,\,x_2)=\left(1-\frac{1}{\sqrt{2}},\frac{1}{2}\right)$ and also at $\left(\frac{1}{2},1-\frac{1}{\sqrt{2}}\right)$.
For $\pi$ supported on two points, we can improve the principal's expected utility to~$\frac{1}{3}$, which is achieved at the distribution that places equal weight on  $\left(1,\frac{1}{3}\right)$ and $\left(\frac{1}{3},1\right)$. Allowing for the third point in the support does not improve the objective.

The optimal information structure $I$ has two signals $L$ and $H$ for both receivers. 
 If $\omega=\ell$, the principal picks an agent uniformly at random and sends the signal $H$. In all other cases (state $\omega=\ell$ and the receiver not picked or state $\omega=h$), the sender sends the signal $L$. The corresponding belief distribution is depicted in Figure~\ref{fig_one_state}. \ed{In both states, there is an agent who attributes belief $2/3$ to the good state, i.e., the optimal structure equalizes the maximum level of optimism in good and bad states. However, in the good state, both players are at this level, and in the bad state, just one.}

\ed{The notion of which state is good and which state is bad can be subjective. {For example, one team of agents \(i \in \{1,2\}\) may consider \(\omega = \ell\) as bad, while another team \(i \in \{3,4\}\) regards \(\omega = h\) as bad.} Supporting morale in both teams in their subjective bad states corresponds to
$v^\omega = \mathbf{1}[\omega = \ell] \cdot \max\{1-x_1, 1-x_2\} + \mathbf{1}[\omega = h] \cdot \max\{x_3, x_4\}.
$
As the correlation of beliefs across teams is irrelevant to the objective, the problem reduces to solving separate one-state persuasion problems for each team.
More generally, if at each state \(\omega\), the sender derives utility from disjoint groups \(N_\omega \subset N\), the persuasion problem reduces to \(|\Omega|\) one-state persuasion problems, each indexed by \(\omega_0 \in \Omega\) and involving receivers in \(N_{\omega_0}\) only. Therefore, supporting morale in the teams \(\{1,2\}\) and \(\{3,4\}\) can be addressed separately using the previously solved one-state problem.
}
\end{example}

\begin{example}(\ed{selective production discouragement in an oligopoly with unknown costs})\label{ex_olygop}
\ed{
    \cite*{mckelvey1986common} consider a Cournot oligopoly, where firms use the same technology.  The technology can be toxic ($\omega=\ell$) or not ($\omega=h$) with prior $p\in (0,1)$ for the low state. Toxic production leads to liability $\gamma_i\cdot  q_i$ where $q_i$ is the amount firm $i$ produced, and $\gamma_i>0$ is a constant. Effectively, this means that the firms face common uncertainty about production costs. 
    Each firm conducts a private experiment about toxicity, resulting in belief $x_i\in [0,1]$ for the state $\ell$. The inverse demand  $D^{-1}(q)=a-b\cdot q$ is linear in total production $q$ and  the production costs of firm $i$ are quadratic~$C_i(q)=c_iq^2/2+d_i q+e_i$. 
    
    \cite*{mckelvey1986common} assume that firms are naive, use their private information only, and do not learn from contemporaneous prices or the production level of the competitor.\footnote{\ed{In other words, firm $i$ sticks to her belief $x_i$ even if it realizes that the competitor's choices indicate a belief that differs from $x_i$, as in the cursed equilibrium of \cite*{eyster2005cursed}.}} Under this assumption, each firm maximizes $D^{-1}(q_1+q_2)\cdot q_i-C_i(q_i)-\gamma_i x_i q_i$ over $q_i$, and the first-order conditions lead to a linear system $(2b+c_i) q_i+ b q_{-i}=a-d_i-\gamma_i x_i$ resulting in linear $q_i(x_i,x_{-i})$, decreasing in $x_i$ and increasing in $x_{-i}$.

We consider a variation of this duopoly model where, instead of private experiments conducted by firms, the government learns $\omega$ by hiring an expert and can then reveal some information to the producers. The product remains valuable for consumers regardless of the technology. However, the firms' production facilities are located differently: 
firm 1 is situated near a residential area, while firm 2 is not. The government aims to maximize social welfare given by
$$u^\omega=\left(q_1+q_2\right)-  \mathbf{1}[\omega=\ell]\cdot  C(q_1),$$
which consists of the total production from both firms minus the pollution cost $C(q_1)$ in the state~$\ell$, increasing in firm~1's production.  Since the quantities $q_i$ are linear in beliefs $(x_1,x_2)$, and the average belief equals the prior, the expected value of $q_i$ is the same across feasible belief distributions. Therefore, the government's problem is equivalent to a one-state persuasion with $v^\ell=-\mathbf{1}[\omega=\ell] \cdot C(q_1)$ and $v^h\equiv 0$.

For convex costs $C$, the objective $v^\ell$ is concave in $(x_1,x_2)$ and thus the optimum in the optimization problem~\eqref{eq_value_as_finite_dim_optimization} from Lemma~\ref{lm_one_state_counting_extreme} is attained at a single point mass. Taking into account the constraints in~\eqref{eq_value_as_finite_dim_optimization}  and the fact that $q_1$ is decreasing in $x_1$ and increasing in $x_2$, the optimal point mass is $\delta_{(1,p)}$. This conditional belief distribution is induced by an information structure that reveals the state to firm~1 while keeping firm~2 uninformed. Although the optimality of such a structure may seem intuitive, we are unaware of an argument that does not rely on an analog of Lemma~\ref{lm_one_state_counting_extreme}.

\begin{figure}[h!]
	\begin{center}
	    \scalebox{0.6}{
	    
		\begin{tikzpicture}[scale=0.5,line width=1.5pt]
		\draw (0,0)--(10,0)--(10,10)--(0,10)--(0,0);

  \draw[thin, densely dashed] (0,3.61)--(10,3.61);		
  		\filldraw[red] (10,10) circle (0.35);
      	\filldraw[red] (10,3.61) circle (0.45);
      	\filldraw[blue] (0,3.61) circle (0.6);
		\node[right] at (0.3,5) {\huge$\frac{1}{2}$};
  \node[left] at (9.7,5) {\Large$0.28$};
  	\node[right] at (0.3,5) {\huge$\frac{1}{2}$};
  \node[left] at (9.7,9) {\Large$0.22$};
		\node[below right] at (10,0) {\Large $x_1(\ell)$};
		\node[above left] at (0,10) {\Large $x_2(\ell)$};
		\node[left] at (-0.6,3.61) {\Large $0.36$};
		\end{tikzpicture}
        }
	\end{center}
	\caption{\ed{The joint distribution of beliefs for Example~\ref{ex_olygop} with pollution cost $C(q_1)=1-\exp(-2q_1)$. The numbers inside the square indicate the probabilities of each outcome, and red/blue colors correspond to $\omega=\ell$ and $\omega=h$, respectively.}\label{fig_duopoly}}
\end{figure}

{For concave cost functions \(C\), the optimal information structure may be less intuitive. When firm~2 assigns a relatively high probability to \(\omega = \ell\), it reduces its production, which encourages increased production by firm~1. This strategic consideration may affect the optimal information structure: in addition to informing firm~1 about the low state (\(x_1 = 1\)), the government may find it optimal to induce \(x_2 < 1/2\) in some scenarios.}

{For example, consider a cost function \(C(q_1) = 1 - \exp(-2q_1)\) and assume the prior \(p=1/2\). Suppose the duopoly has parameters \(a = 21\), \(b = 2\), \(c_1 = c_2 = 0\), \(d_1 = 9\), \(d_2 = 3\), \(\gamma_1 = 3\), and \(\gamma_2 = 12\), which results in the production function \(q_1(x_1, x_2) = 1 - x_1 + 2x_2\). The optimal distribution in Lemma~\ref{lm_one_state_counting_extreme} is then supported on two points \((1,1)\) and \((1, t)\) with some \(t \leq 1/2\). The constraint~\eqref{eq_admissible_marginals_omega0} binds for firm~2, resulting in weights \(\frac{1 - 2t}{1 - t}\) and \(\frac{t}{1 - t}\), respectively.} {The government's objective reduces to maximizing \(\frac{t}{1 - t} \left( \exp\left[4(1 - t)\right] - 1 \right)\), which leads to the optimal \(t^* \approx 0.3608\).}
The resulting belief distribution is depicted in Figure~\ref{fig_duopoly}. {This distribution is induced by an information structure that reveals the state to firm~1 and sends a noisy signal \(s_2 \in \{L, H\}\) to firm~2; \(s_2 = L\) is sent with certainty in the low state, and in the high state, \(s_2 = L\) and \(s_2 = H\) are sent with probabilities approximately \(0.564\) and \(0.436\).}

}

\end{example}

\subsection{Supermodular Persuasion}\label{sect_super}

\ed{
Supermodular objectives are omnipresent in economics, arising in contexts where agents' actions are complements. In this section, we demonstrate that symmetric multi-receiver persuasion problems with supermodular objectives can be reduced to single-receiver problems, regardless of the number of agents and states.

Recall that a function $G\colon \mathbb{R}^n \to \mathbb{R}$ is \emph{supermodular} if, for all $z, z' \in \mathbb{R}^n$,
\[
G(z \lor z') + G(z \land z') \geq G(z) + G(z'),
\]
where $z \lor z'$ and $z \land z'$ denote the component-wise maximum and minimum of $z$ and $z'$, respectively. 
Informally, increasing one variable makes increasing another variable more beneficial.
For twice continuously differentiable functions, this condition is equivalent to the requirement that all mixed partial derivatives $\partial_{z_i}\partial_{z_j} G$ are non-negative for all $i \ne j$.

For arbitrary finite sets of agents $N = \{1,2,\ldots, n\}$ and states $\Omega$, we call a persuasion problem \emph{supermodular}\footnote{\ed{Our setting should not be confused with that of papers on information design in supermodular games---e.g., \cite*{halac2022addressing} and \cite*{morris2024implementation}---which consider receivers with binary actions and model externalities among them through supermodular utilities, emphasizing strategic interactions. In contrast, our first-order persuasion setting allows for arbitrarily rich actions but rules out strategic externalities; here, supermodularity refers to the sender's objective function instead.}
} if
\begin{equation}\label{eq_super_modular_def}
    v^\omega(x_1,\ldots, x_n) = G^\omega\Big(a_1^\omega(x_1),\ldots, a_n^\omega(x_n)\Big),
\end{equation}
where $G^\omega\colon \mathbb{R}^n \to \mathbb{R}$ is supermodular, and $a_i^\omega\colon \Delta(\Omega) \to \mathbb{R}$. As usual, we assume that $v^\omega$ is upper semicontinuous.

We say that a supermodular problem is \emph{agent-symmetric} if the functions $a_i^\omega$ do not depend on~$i$ (thus $a_i^\omega=a^\omega$) and $G^\omega(z_1,\ldots, z_n) = G^\omega(z_{\sigma(1)},\ldots, z_{\sigma(n)})$ for any permutation $\sigma$ of $\{1,\ldots,n\}$.

\begin{lemma}\label{lm_supermodular}
An agent-symmetric supermodular persuasion problem $B$ with utility~\eqref{eq_super_modular_def} is equivalent to a single-receiver problem $\overline{B}$ that has the same prior $p$ and a state-independent utility
\[
\overline{v}(x) = \sum_{\omega\in \Omega} x(\omega)\cdot G^{\omega}\big(a^\omega(x),\ldots, a^\omega(x)\big),\quad x\in \Delta(\Omega).
\]
Namely,
\[
\val[B] = \val\big[\overline{B}\big] = \cav\big[\overline{v}\big](p),
\]
where $\cav\big[\overline{v}\big]$ denotes the concavification of $\overline{v}$ over $\Delta(\Omega)$. This value in the original problem can be achieved by an information structure with $|\Omega|$ public signals.
\end{lemma}

The proof of Lemma~\ref{lm_supermodular} is presented in Appendix~\ref{app_supermodular_proof}. 
The key idea is leveraging the connection between persuasion and optimal transport (Corollary~\ref{cor_persuasion_as_transport}) and the fact that optimal transportation problems with supermodular objectives have particularly simple solutions, given by \emph{assortative} matching.
Indeed, let $\tau_i \in \Delta(\mathbb{R})$ be probability distributions with bounded support. Denote by $f_i$ the quantile function (inverse cumulative distribution function) of $\tau_i$, i.e.,
$
f_i(t) = \inf \{x \in \mathbb{R} \mid \tau_i((-\infty, x]) \geq t\},$ $t \in (0,1].
$
Then, the optimal value for a transportation problem with a supermodular objective $G$ is given by
\[
MK_G[\tau_1,\ldots, \tau_n] = \int_{0}^{1} G\big(f_1(t),\ldots, f_n(t)\big)\, \dd t.
\]
It is attained at the distribution of the vector $(f_1(t),\ldots, f_n(t))$ for $t$ uniformly distributed on $[0,1]$, the so-called \emph{assortative matching} of $\tau_1,\ldots, \tau_n$. Particular $2$-marginal versions of this result have appeared in economics since \cite*{becker1973theory}, but the result holds in great generality for any $n$,  requiring only the measurability of $G$ and no additional regularity assumptions  \citep*{burchard2006rearrangement}.

To prove Lemma~\ref{lm_supermodular}, we note that in an agent-symmetric problem, it suffices to consider marginal distributions $\lambda_i^\omega$ that are the same for all agents. Due to the supermodularity of   $G^\omega$, assortative matching of 
$a^\omega(x_1),\ldots, a^\omega(x_n)$ 
maximizes the sender's utility in each state~$\omega$.
Since the marginal distributions are identical, assortative matching leads to $x_1 = \ldots = x_n$. 
Consequently, providing the same information to all agents via public signals achieves the optimal outcome. This effectively reduces the multi-receiver persuasion problem to a single-receiver problem, completing the argument.

\begin{example}[effort in teams]\label{ex_effort_in_teams}
Suppose each receiver~$i$ chooses an action $a_i(x_i)$ based on their belief~$x_i$---e.g., a worker selects the effort level depending on the belief about the project's characteristics. Suppose the production function $G(a_1,\ldots, a_n)$ is supermodular---meaning that the marginal value of each worker's effort increases with the effort of others. Assume that workers are similar in how their beliefs affect their actions; that is, the functions $a_i(x)$ are identical across all workers.
Then, Lemma~\ref{lm_supermodular} implies that the informed team leader can achieve the best outcome by communicating with the team publicly. 
\end{example}

The persuasion literature often deals with finite state spaces—mostly binary—or continuous states where only posterior means matter. While our approach focuses on small state spaces, Lemma~\ref{lm_supermodular} also applies when the sender's utility depends on receivers' posterior means.

\begin{example}[one-dimensional state and mean-measurable objectives]\label{ex_one_dim_state}

Consider the state space $\Omega = \{\omega_1, \ldots, \omega_k\} \subset \mathbb{R}$.
Let $a_i(x_i)$ be receiver~$i$'s posterior mean, i.e., 
$
a_i(x_i) = \E_I[\omega \mid x_i]
$.
Suppose the sender's utility is
\[
v^\omega(x_1, \ldots, x_n) = G^\omega\big(\E_I[\omega \mid x_i], \ldots, \E_I[\omega \mid x_n]\big),
\]
where $G^\omega$ is supermodular and symmetric. 

For a concrete illustration, consider a firm (the sender) selling a product of uncertain quality \(\omega \in \mathbb{R}\) to multiple markets (the receivers). In each market, the total amount purchased is a function of the consumers' posterior mean \(\E_I[\omega \mid x_i]\). The firm's profit depends supermodularly on the entire demand profile---due to factors like economies of scale.
By Lemma~\ref{lm_supermodular}, optimal persuasion reduces to using public signals. The same conclusion holds if $a_i$ are based on other statistics of posterior beliefs, such as posterior medians or specific quantiles as in \cite*{yang2024monotone}. 
\end{example}

}

\ed{The following example demonstrates that, despite the apparent simplicity, supermodular objectives can lead to rich information transmission patterns even in binary-state two-receiver problems.
\begin{example}[revenue maximization by a public option provider]

Consider a government agency or nonprofit organization that provides a public option, such as affordable housing or public transportation, to two distinct markets. The quality of the service is uncertain---it can be either high ($\omega = h$) or low ($\omega = \ell$)---with a prior probability $p \in [0,1]$ of being low. Consumers in each market are uncertain about the service's quality and decide whether to use it based on their beliefs. The provider charges a regulated, exogenously fixed price $\alpha$ for the service in both markets and must meet the total demand. Consequently, the provider cannot rely on price mechanisms to influence demand for the service and instead relies on informational tools, affecting each market's beliefs about quality through private signals.

Let $F_i(x_i)$ be a decreasing function that denotes the demand from market~$i$ when consumers have a belief $x_i \in [0,1]$ that the quality is low. The provider's cost of supplying the service is given by a cost function $c^\omega(q_1, q_2)$, where $q_i$ is the quantity demanded in market~$i$. The cost $c^\omega$ is submodular in quantities, i.e., the marginal cost of supplying one market decreases when supplying both markets together. The provider's total profit
\[
v^\omega(x_1, x_2) = \alpha \cdot \big(F_1(x_1) + F_2(x_2)\big) - c^\omega\big(F_1(x_1), F_2(x_2)\big)
\]
is supermodular in beliefs $(x_1, x_2)$. Lemma~\ref{lm_supermodular} reduces the two-receiver private persuasion to persuading the receivers via public signals. In particular, the provider does not need to worry about information leakage between markets.

As a concrete example, consider $F_i(x_i)=1-x_i$, $c^h(q_1,q_2)=\sqrt{q_1+q_2}$, and $c^\ell(q_1,q_2) = \frac{1}{3}\sqrt{q_1+q_2}$. Hence, the indifference points between buying and not buying are distributed uniformly in each of the markets, and providing low-quality service is thrice less costly than providing high-quality service.
The auxiliary state-independent utility function becomes $$\overline{v}(x)=2\alpha(1-x)-\left(1-\frac{2}{3}x\right)\sqrt{2-2x}.$$
Applying concavification to $\overline{v}$, we find that for $p \in \left[0, \tfrac{1}{2}\right]$---that is, when high quality is more likely---it is optimal to reveal no information. For $p \in \left[\tfrac{1}{2}, 1\right]$, it is optimal to mix between revealing low quality ($x = 1$) and inducing the maximally uncertain belief ($x = \tfrac{1}{2}$) thus avoiding costly situations with small but positive demand and ensuring that demand is either zero or substantial.

\end{example}
}

\subsection{Constructing Solutions via Duality}\label{sect_polarization}

\ed{In this section, we use the dual representation (Proposition~\ref{prop_dual value representation}) to solve persuasion problems. The dual approach is particularly convenient when we guess a specific information structure~$I$ is optimal. To verify this guess, we construct---guided by the complementary slackness conditions---a feasible solution to the dual problem that yields the same value as~$I$. The existence of such a dual solution confirms the optimality of $I$ since any feasible solution to the dual gives an upper bound on the value. 

We demonstrate this general technique in Appendix~\ref{app_dual}, where we derive a condition for the optimality of information structures revealing no information to one receiver and partial information to the other in two-receiver problems. While our duality approach is especially useful for problems where we anticipate a simple optimal information structure, it can also be applied to problems where the optimal structure is highly nontrivial \citep*{kravchenko2024}.

Here, we focus on two illustrative examples.
As a first step, we show that the dual from Proposition~\ref{prop_dual value representation} simplifies in the case of two receivers and a binary state.} Identifying each belief \( x \in \Delta(\Omega) \) with \( x(\ell) \in [0,1] \), we exploit a unique feature of the binary-state case: the last condition in~\eqref{eq:dual_value_explicit} uniquely determines \( \varphi_i^h \) given \( \varphi_i^\ell \). This reduces the number of functions from four to two. Let us denote \( \alpha_i(x) = \frac{\varphi_i^\ell(x)}{1 - x} \). Therefore, $\varphi_i^h(x)=-x\cdot\alpha_i(x)$ and so $\alpha_i$ is not singular at~$x=1$ and thus continuous. We conclude that~\eqref{eq:dual_value_explicit} reduces to 
\begin{align}\label{eq:dual_value_2_receivers_simplified}
			\val[B]=&\inf_{{\footnotesize \begin{array}{rl}  V^\omega\in\R, & \mbox{continuous  $\alpha_i$ on $[0,1]$ such that}\\
			v^\ell(x_1,x_2)&\leq V^\ell+(1-x_1)\alpha_1(x_1)+ (1-x_2)\alpha_2(x_2)
			\\
			v^h(x_1,x_2)&\leq V^h-x_1\cdot\alpha_1(x_1) -x_2\cdot\alpha_2(x_2)
			\end{array}}}  p\cdot V^\ell+(1-p)V^h.
		\end{align}
\ed{We demonstrate how this simplified dual can be applied in the following examples.}

\begin{example}[$\beta$-polarization]\label{ex_x_to_third}
\ed{Consider a pair of agents with beliefs $x_1,x_2\in [0,1]$ about a binary state induced by some information structure~$I$. Polarization of beliefs can be quantified using various measures. One such measure is $\E_I[|x_1 - x_2|^\beta],$ which we refer to as the \emph{$\beta$-polarization}. Larger values of $\beta$ place more emphasis on significant differences in beliefs, making the measure less sensitive to small discrepancies.

Assuming a prior $p = {1}/{2}$, we aim to determine the maximal possible $\beta$-polarization. We show that for every $\beta$ in the interval $\left(0, \beta_{\max}\right]$, where $\beta_{\max} \approx 2.25751$, the maximal $\beta$-polarization is achieved by an information structure that reveals the state fully to one receiver while keeping the other receiver completely uninformed. This information structure yields a tight upper bound of~$2^{-\beta}$ on $\beta$-polarization.

To see this, consider the persuasion problem $B$ with the objective $v^\omega(x_1, x_2) = |x_1 - x_2|^\beta$ and utilize the dual~\eqref{eq:dual_value_2_receivers_simplified}. We choose $V^\ell = V^h = 2^{-\beta}$ and define the functions $\alpha_i(x)$ as
\begin{equation*}\label{eq_alpha}
\alpha_1(x)=\alpha_2(x)=\left\{\begin{array}{cc}
   \frac{(1-x)^\beta-2^{-\beta}}{1-x}, & x\leq 1/2\\ & \\ 
\frac{2^{-\beta}-x^\beta}{x}, & x\geq 1/2
\end{array}\right..
\end{equation*}
An elementary computation shows that, with these choices, the constraints in~\eqref{eq:dual_value_2_receivers_simplified} are satisfied.\footnote{This can be verified numerically; see Appendix~\ref{app_code} for the Mathematica code.} Therefore, for every information structure $I$, we have
\[
\E_I[|x_1 - x_2|^\beta] \leq \val[B] \leq p V^\ell + (1 - p) V^h = \dfrac{1}{2} \cdot 2^{-\beta} + \dfrac{1}{2} \cdot 2^{-\beta} = 2^{-\beta}.
\]
As mentioned, this bound is attained by revealing the state fully to one of the agents and keeping the other agent completely uninformed. 
This tight bound on the $\beta$-polarization generalizes previous results obtained by \cite*{burdzy2020bounds} for $\beta = 1$ and \cite*{arieli2020feasible} for $\beta \in (0,2]$.

One may wonder how we arrived at the specific form of $\alpha$. In Appendix~\ref{app_dual}, we provide the intuition behind this choice and present Proposition~\ref{prop_sufficient_full_info_no_info_h}, which describes a class of persuasion problems---including the maximal $\beta$-polarization problem---for which a similar construction applies.
}

It is important to note that while the full-information/no-information policy induces a belief difference of $|x_1 - x_2| = {1}/{2}$ with probability one, the sender could, in principle, induce larger belief differences. However, there is a trade-off between the magnitude of the belief differences and the probability of their occurrence. For $\beta \leq \beta_{\max}$, the sender's relative benefit from inducing larger belief discrepancies does not outweigh the loss in probability. For higher values of $\beta$, alternative information structures yield higher $\beta$-polarization than the full-information/no-information policy, e.g., the information structure from Example~\ref{ex_one_state} results in a higher value for $\beta \geq 2.41$. \ed{The asymptotic behavior as $\beta \to +\infty$ has been recently studied by \cite*{cichomski2023coherent}, and the complete solution for $\beta > \beta_{\max}$ has been obtained by \cite*{kravchenko2024} using a variant of our duality approach.}

\end{example}

\ed{
\begin{example}[profit maximization by an informed retailer]\label{retailer}

Consider a retailer who is informed of the actual value of a good---e.g., an antiquarian or a two-sided platform---and aims to convince the owner that the value is low while convincing a potential buyer that it is high. For simplicity, suppose the value of the good is $1$ if $\omega = h$ and $0$ if $\omega = \ell$. In contrast to Example~\ref{ex_x_to_third}, we allow for arbitrary prior probability $p \in (0,1)$ of~$\omega = \ell$.

Agent~1 is the owner of the good and is willing to sell it for a price equal to her expected value $\E_I[\omega \mid x_1] = 1 - x_1$. Similarly, the buyer is willing to pay $\E_I[\omega \mid x_2] = 1 - x_2$. The retailer can profit from the transaction if $x_1 > x_2$, collecting an expected profit of $\E_I[(x_1 - x_2) \cdot \mathbf{1}[x_1 \geq x_2]]$.
This objective is equivalent to the $1$-polarization objective from Example~\ref{ex_x_to_third} because
\[
(x_1 - x_2) \cdot \mathbf{1}[x_1 \geq x_2] = \frac{|x_1 - x_2|}{2} - \frac{x_1 - x_2}{2}
\]
and $\E_I[x_1 - x_2] = 0$ for any information structure by the martingale property. Therefore, maximizing the retailer's profit is equivalent to maximizing $\E_I[|x_1 - x_2|]$, i.e., $1$-polarization. The 1-polarization problem has been solved by \cite*{burdzy2020bounds} for any prior. Here, we provide a simple alternative solution via duality.

We argue that the maximal $1$-polarization is achieved by an information structure that reveals the state to one receiver while keeping the other uninformed, yielding a tight upper bound on $1$-polarization of $2p(1 - p)$. To see this, consider the dual problem~\eqref{eq:dual_value_2_receivers_simplified} with $V^\ell = 2(1 - p)^2$ and $V^h = 2p^2$, and define\footnote{\ed{The heuristics behind the choice function $\alpha$ is explained in Proposition~\ref{prop_sufficient_full_info_partial_info_optimality} from Appendix~\ref{app_dual}.}}
\[
\alpha_1(x)=\alpha_2(x) = \left\{
\begin{array}{ll}
1 - \frac{2(1 - p)^2}{1 - x}, & \text{if } x \leq p, \\
\frac{2p^2}{x} - 1, & \text{if } x \geq p.
\end{array}
\right.
\]
A straightforward computation shows that the constraints in~\eqref{eq:dual_value_2_receivers_simplified} are satisfied. 
Thus
\[
\val[B] \leq p V^\ell + (1 - p) V^h = 2p(1 - p),
\]
confirming the optimality of the full-information/no-information policy. 

Returning to the retailer's problem, we conclude that it is optimal for the retailer to reveal the state to one side of the market while keeping the other uninformed. Note that, similarly to Example~\ref{ex_olygop}, the buyer and seller are naive in the sense of \cite*{eyster2005cursed} as they do not account for the information revealed by the occurrence of the transaction itself. This allows the retailer to obtain positive profit despite the no-trade theorem. 
\end{example}
}

\ed{

\section{Conclusion}\label{sec:conclude}

 There is potential for generalizing our approach---some extensions are straightforward, while others are not. We outline several directions that we find particularly interesting.

\paragraph{Infinite state spaces and uncommon priors.} For simplicity, we have assumed a finite number of states and that all receivers share a common prior. Extending our results to countably infinite or uncountable state spaces—such as $\Omega = [0,1]$ or a compact metric space—poses no conceptual difficulties but requires additional technical considerations, like ensuring that the conditional belief distributions $\mu^\omega$ are measurable in $\omega$. Similarly, our results can be extended to settings with uncommon priors among receivers. Indeed, if the true state is distributed according to $p$, but receiver $i$ believes the distribution is $p' \ne p$, then any belief $x'$ induced by an information structure can be obtained by a simple likelihood re-weighting of the belief $x$ that would be held if the receiver's prior were $p$. Therefore, feasibility and persuasion problems with uncommon priors can be reduced to the common prior case.

\paragraph{Posterior means.} Extending to infinite state spaces is straightforward if we consider the entire beliefs $x_i \in \Delta(\Omega)$. However, it is common to focus on posterior means $m_i = \E_I[\omega \mid x_i]$ when the state $\omega$ is real-valued. While our approach offers some insights for such settings (see Example~\ref{ex_one_dim_state}), we are not aware of any simple characterization of feasible distributions of $(m_1, \ldots, m_n)$, whether conditional or unconditional. Although a version of Theorem~\ref{th_feasibility} applies to $(m_1, \ldots, m_n)$---correlation  of feasible marginal distributions $m_i$ conditional on the realized state can be arbitrary---to our knowledge, there is no characterization of feasible one-agent marginals analogous to Observation~\ref{obs_one_receiver_feasibility}. This missing ingredient is crucial for making an analog of Theorem~\ref{th_feasibility} useful in this context.

\paragraph{Algorithmic implications.} Representing the persuasion problem as a maximization over feasible conditional belief distributions leads to primal and dual linear programming formulations. While we have not explored this direction in detail, both formulations appear promising from a numerical perspective. Although the optimal distributions may have uncountable support \citep*{cichomski2025existence}, there always exists an approximately optimal feasible conditional distribution supported on a sufficiently fine grid. On this grid, both the primal and dual problems become finite-dimensional and can be solved using standard LP solvers. Furthermore, we anticipate that the dual problem may offer computational advantages. Indeed, in the primal problem, we maximize a linear objective over feasible distributions. By focusing on conditional distributions, we increase the dimension by a factor of $|\Omega|$ (since we have $|\Omega|$ distributions instead of one), resulting in a linear increase with the number of states. However, the dimension of the unconditional problem can already be prohibitive unless the number of receivers $|N|$ is small. Indeed, when each receiver's belief space is discretized into $D$ points, the set of all unconditional distributions lies in an $D^{|N|}$-dimensional space, which grows exponentially with the number of receivers. In contrast, the dual problem involves approximately $|\Omega| \cdot D \cdot |N|$ variables, making it potentially computationally tractable even for large numbers of receivers.

\paragraph{Sets of belief distributions with a simple structure of extreme points.} Certain classes of persuasion problems are tractable because the relevant sets of belief distributions have simple extreme points—for example, in the single-receiver model, mean-measurable, and quantile-measurable persuasion problems \citep*{kamenica2011bayesian,kleiner2021extreme,arieli2023optimal,yang2024monotone}. Our analysis of one-state persuasion (Section~\ref{sect_onestate}) adds to this list by identifying a projection of the set of all feasible conditional distributions that has simple extreme points. This projection is particularly intriguing because, prior to projection, the set can have complex extreme distributions with uncountable support. Finding other economically relevant projections with simple extreme points can bring tractability to new classes of multi-receiver problems.

}

\bibliography{main}

\begin{thebibliography}{75}
\providecommand{\natexlab}[1]{#1}
\providecommand{\url}[1]{\texttt{#1}}
\expandafter\ifx\csname urlstyle\endcsname\relax
  \providecommand{\doi}[1]{doi: #1}\else
  \providecommand{\doi}{doi: \begingroup \urlstyle{rm}\Url}\fi

\bibitem[Alonso and C{\^a}mara(2016)]{alonso2016bayesian}
R.~Alonso and O.~C{\^a}mara.
\newblock Bayesian persuasion with heterogeneous priors.
\newblock \emph{Journal of Economic Theory}, 165:\penalty0 672--706, 2016.

\bibitem[Arieli and Babichenko(2019)]{arieli2019private}
I.~Arieli and Y.~Babichenko.
\newblock Private bayesian persuasion.
\newblock \emph{Journal of Economic Theory}, 182:\penalty0 185--217, 2019.

\bibitem[Arieli and Babichenko(2022)]{arieli2022population}
I.~Arieli and Y.~Babichenko.
\newblock A population's feasible posterior beliefs.
\newblock In \emph{Proceedings of the 23rd ACM Conference on Economics and
  Computation}, pages 326--327, 2022.

\bibitem[Arieli et~al.(2021{\natexlab{a}})Arieli, Babichenko, Sandomirskiy, and
  Tamuz]{arieli2020feasible}
I.~Arieli, Y.~Babichenko, F.~Sandomirskiy, and O.~Tamuz.
\newblock Feasible joint posterior beliefs.
\newblock \emph{Journal of Political Economy}, 129\penalty0 (9):\penalty0
  2546--2594, 2021{\natexlab{a}}.

\bibitem[Arieli et~al.(2021{\natexlab{b}})Arieli, Babichenko, Sandomirskiy, and
  Tamuz]{arieli2021feasible}
I.~Arieli, Y.~Babichenko, F.~Sandomirskiy, and O.~Tamuz.
\newblock Feasible joint posterior beliefs (through examples).
\newblock \emph{ACM SIGecom Exchanges}, 19\penalty0 (1):\penalty0 21--29,
  2021{\natexlab{b}}.

\bibitem[Arieli et~al.(2022)Arieli, Babichenko, and
  Sandomirskiy]{arieli2022persuasion}
I.~Arieli, Y.~Babichenko, and F.~Sandomirskiy.
\newblock Persuasion as transportation.
\newblock In \emph{Proceedings of the 23rd ACM Conference on Economics and
  Computation, New York, NY, USA: Association for Computing Machinery, EC},
  volume~22, page 468, 2022.

\bibitem[Arieli et~al.(2023)Arieli, Babichenko, Smorodinsky, and
  Yamashita]{arieli2023optimal}
I.~Arieli, Y.~Babichenko, R.~Smorodinsky, and T.~Yamashita.
\newblock Optimal persuasion via bi-pooling.
\newblock \emph{Theoretical Economics}, 18\penalty0 (1):\penalty0 15--36, 2023.

\bibitem[Aumann(1976)]{aumann1976agreeing}
R.~J. Aumann.
\newblock Agreeing to disagree.
\newblock \emph{The Annals of Statistics}, pages 1236--1239, 1976.

\bibitem[Aumann and Maschler(1995)]{aumann1995repeated}
R.~J. Aumann and M.~Maschler.
\newblock \emph{Repeated games with incomplete information}.
\newblock MIT press, 1995.
\newblock In collaboration with {Richard} {E}. {S}tearns.

\bibitem[Becker(1973)]{becker1973theory}
G.~S. Becker.
\newblock A theory of marriage: Part i.
\newblock \emph{Journal of Political economy}, 81\penalty0 (4):\penalty0
  813--846, 1973.

\bibitem[Bergemann and Morris(2016)]{bergemann2016bayes}
D.~Bergemann and S.~Morris.
\newblock Bayes correlated equilibrium and the comparison of information
  structures in games.
\newblock \emph{Theoretical Economics}, 11\penalty0 (2):\penalty0 487--522,
  2016.

\bibitem[Bergemann and Morris(2019)]{bergemann2019information}
D.~Bergemann and S.~Morris.
\newblock Information design: A unified perspective.
\newblock \emph{Journal of Economic Literature}, 57\penalty0 (1):\penalty0
  44--95, 2019.

\bibitem[Blackwell(1951)]{blackwell1951comparison}
D.~Blackwell.
\newblock Comparison of experiments.
\newblock In \emph{Proceedings of the Second Berkeley Symposium on Mathematical
  Statistics and Probability}, pages 93--102. University of California Press,
  1951.

\bibitem[Boerma et~al.(2021)Boerma, Tsyvinski, and Zimin]{boerma2021sorting}
J.~Boerma, A.~Tsyvinski, and A.~P. Zimin.
\newblock Sorting with team formation.
\newblock Technical report, National Bureau of Economic Research, 2021.

\bibitem[Bogachev and Kolesnikov(2012)]{BoKo}
V.~I. Bogachev and A.~V. Kolesnikov.
\newblock The {Monge}--{Kantorovich} problem: achievements, connections, and
  perspectives.
\newblock \emph{Russian Math. Surveys}, 67\penalty0 (5):\penalty0 785--890,
  Oct. 2012.

\bibitem[Brooks et~al.(2022)Brooks, Frankel, and
  Kamenica]{brooks2019information}
B.~Brooks, A.~Frankel, and E.~Kamenica.
\newblock Information hierarchies.
\newblock \emph{Econometrica}, 90\penalty0 (5):\penalty0 2187--2214, 2022.

\bibitem[Burchard and Hajaiej(2006)]{burchard2006rearrangement}
A.~Burchard and H.~Hajaiej.
\newblock Rearrangement inequalities for functionals with monotone integrands.
\newblock \emph{Journal of Functional Analysis}, 233\penalty0 (2):\penalty0
  561--582, 2006.

\bibitem[Burdzy and Pal(2021)]{burdzy2021can}
K.~Burdzy and S.~Pal.
\newblock Can coherent predictions be contradictory?
\newblock \emph{Advances in Applied Probability}, 53\penalty0 (1):\penalty0
  133--161, 2021.

\bibitem[Burdzy and Pitman(2020)]{burdzy2020bounds}
K.~Burdzy and J.~Pitman.
\newblock Bounds on the probability of radically different opinions.
\newblock \emph{Electronic Communications in Probability}, 25:\penalty0 1--12,
  2020.

\bibitem[Carlier(2012)]{Carlier}
G.~Carlier.
\newblock Optimal transportation and economic applications.
\newblock \emph{Lecture Notes}, 2012.

\bibitem[Chiappori et~al.(2010)Chiappori, McCann, and
  Nesheim]{chiappori2010hedonic}
P.-A. Chiappori, R.~J. McCann, and L.~P. Nesheim.
\newblock Hedonic price equilibria, stable matching, and optimal transport:
  equivalence, topology, and uniqueness.
\newblock \emph{Economic Theory}, pages 317--354, 2010.

\bibitem[Cichomski(2020)]{cichomski2020maximal}
S.~Cichomski.
\newblock Maximal spread of coherent distributions: a geometric and
  combinatorial perspective.
\newblock \emph{arXiv preprint arXiv:2007.08022}, 2020.

\bibitem[Cichomski and
  Os{\k{e}}kowski(2022{\natexlab{a}})]{cichomski2022contradictory}
S.~Cichomski and A.~Os{\k{e}}kowski.
\newblock Contradictory predictions with multiple agents.
\newblock \emph{arXiv preprint arXiv:2211.02446}, 2022{\natexlab{a}}.

\bibitem[Cichomski and Os{\k{e}}kowski(2022{\natexlab{b}})]{cichomski2022doob}
S.~Cichomski and A.~Os{\k{e}}kowski.
\newblock Doob's estimate for coherent random variables and maximal operators
  on trees.
\newblock \emph{arXiv preprint arXiv:2211.02434}, 2022{\natexlab{b}}.

\bibitem[Cichomski and Os{\k{e}}kowski(2023)]{cichomski2023coherent}
S.~Cichomski and A.~Os{\k{e}}kowski.
\newblock Coherent distributions on the square --- extreme points and
  asymptotics.
\newblock \emph{arXiv preprint arXiv:2305.09547}, 2023.

\bibitem[Cichomski and Os{\k{e}}kowski(2025)]{cichomski2025existence}
S.~Cichomski and A.~Os{\k{e}}kowski.
\newblock On the existence of extreme coherent distributions with no atoms.
\newblock \emph{Journal of Theoretical Probability}, 38\penalty0 (1):\penalty0
  1--15, 2025.

\bibitem[Cichomski and Os\k{e}kowski(2021)]{cichomski2021maximal}
S.~Cichomski and A.~Os\k{e}kowski.
\newblock The maximal difference among expert’s opinions.
\newblock \emph{Electronic Journal of Probability}, 26:\penalty0 1--17, 2021.

\bibitem[Cichomski and Petrov(2023)]{cichomski2023combinatorial}
S.~Cichomski and F.~Petrov.
\newblock A combinatorial proof of the burdzy--pitman conjecture.
\newblock \emph{Electronic Communications in Probability}, 28:\penalty0 1--7,
  2023.

\bibitem[Cieslak et~al.(2021)Cieslak, Malamud, Schrimpf,
  et~al.]{cieslak2021optimal}
A.~Cieslak, S.~Malamud, A.~Schrimpf, et~al.
\newblock Optimal transport of information.
\newblock Technical report, CEPR Discussion Papers, 2021.

\bibitem[Daskalakis et~al.(2017)Daskalakis, Deckelbaum, and
  Tzamos]{daskalakis2017strong}
C.~Daskalakis, A.~Deckelbaum, and C.~Tzamos.
\newblock Strong duality for a multiple-good monopolist.
\newblock \emph{Econometrica}, 85\penalty0 (3):\penalty0 735--767, 2017.

\bibitem[Dawid et~al.(1995)Dawid, DeGroot, and Mortera]{dawid1995coherent}
A.~Dawid, M.~DeGroot, and J.~Mortera.
\newblock Coherent combination of experts' opinions.
\newblock \emph{Test}, 4\penalty0 (2):\penalty0 263--313, 1995.

\bibitem[Dizdar and Kov{\'a}{\v{c}}(2020)]{dizdar2020simple}
D.~Dizdar and E.~Kov{\'a}{\v{c}}.
\newblock A simple proof of strong duality in the linear persuasion problem.
\newblock \emph{Games and Economic Behavior}, 122:\penalty0 407--412, 2020.

\bibitem[Doval and Smolin(2024)]{doval2024persuasion}
L.~Doval and A.~Smolin.
\newblock Persuasion and welfare.
\newblock \emph{Journal of Political Economy}, 132\penalty0 (7):\penalty0
  000--000, 2024.

\bibitem[Dworczak and Kolotilin(2019)]{dworczak2019persuasion}
P.~Dworczak and A.~Kolotilin.
\newblock The persuasion duality.
\newblock \emph{arXiv preprint arXiv:1910.11392}, 2019.

\bibitem[Dworczak and Martini(2019)]{dworczak2019simple}
P.~Dworczak and G.~Martini.
\newblock The simple economics of optimal persuasion.
\newblock \emph{Journal of Political Economy}, 127\penalty0 (5):\penalty0
  1993--2048, 2019.

\bibitem[Ekeland(2010)]{ekeland2010notes}
I.~Ekeland.
\newblock Notes on optimal transportation.
\newblock \emph{Economic Theory}, pages 437--459, 2010.

\bibitem[Eyster and Rabin(2005)]{eyster2005cursed}
E.~Eyster and M.~Rabin.
\newblock Cursed equilibrium.
\newblock \emph{Econometrica}, 73\penalty0 (5):\penalty0 1623--1672, 2005.

\bibitem[Figalli et~al.(2011)Figalli, Kim, and
  McCann]{figalli2011multidimensional}
A.~Figalli, Y.-H. Kim, and R.~J. McCann.
\newblock When is multidimensional screening a convex program?
\newblock \emph{Journal of Economic Theory}, 146\penalty0 (2):\penalty0
  454--478, 2011.

\bibitem[Galichon(2016)]{galichon2016optimal}
A.~Galichon.
\newblock \emph{Optimal transport methods in economics}.
\newblock Princeton University Press, 2016.

\bibitem[Galichon(2021)]{galichon2021survey}
A.~Galichon.
\newblock A survey of some recent applications of optimal transport methods to
  econometrics.
\newblock \emph{arXiv preprint arXiv:2102.01716}, 2021.

\bibitem[Galperti and Perego(2018)]{galperti2018dual}
S.~Galperti and J.~Perego.
\newblock A dual perspective on information design.
\newblock \emph{Available at SSRN 3297406}, 2018.

\bibitem[Galperti et~al.(2023)Galperti, Levkun, and Perego]{galperti2023value}
S.~Galperti, A.~Levkun, and J.~Perego.
\newblock The value of data records.
\newblock \emph{Review of Economic Studies}, page rdad044, 2023.

\bibitem[Guillen and McCann(2013)]{McCannGuill}
N.~Guillen and R.~McCann.
\newblock Five lectures on optimal transportation: Geometry, regularity and
  applications.
\newblock In \emph{Analysis and Geometry of Metric Measure Spaces: Lecture
  Notes of the 50th S{\'e}minaire de Math{\'e}matiques Sup{\'e}rieures (SMS),
  Montr{\'e}al, 2011}, chapter~6, pages 145--180. CRM Proceedings \& Lecture
  Notes, 2013.
\newblock \doi{10.1090/crmp/056/06}.

\bibitem[Gutmann et~al.(1991)Gutmann, Kemperman, Reeds, and
  Shepp]{gutmann1991existence}
S.~Gutmann, J.~Kemperman, J.~Reeds, and L.~A. Shepp.
\newblock Existence of probability measures with given marginals.
\newblock \emph{The Annals of Probability}, pages 1781--1797, 1991.

\bibitem[Halac et~al.(2022)Halac, Lipnowski, and
  Rappoport]{halac2022addressing}
M.~Halac, E.~Lipnowski, and D.~Rappoport.
\newblock Addressing strategic uncertainty with incentives and information.
\newblock In \emph{AEA Papers and Proceedings}, volume 112, pages 431--437.
  American Economic Association 2014 Broadway, Suite 305, Nashville, TN 37203,
  2022.

\bibitem[He et~al.(2021)He, Sandomirskiy, and Tamuz]{he2021private}
K.~He, F.~Sandomirskiy, and O.~Tamuz.
\newblock Private private information.
\newblock \emph{arXiv preprint arXiv:2112.14356}, 2021.

\bibitem[Herings et~al.(2020)Herings, Karos, and Kerman]{herings2020belief}
P.~Herings, D.~Karos, and T.~Kerman.
\newblock Belief inducibility and informativeness.
\newblock \emph{GSBE Research Memorandum 20/027}, 2020.

\bibitem[Kamenica and Gentzkow(2011)]{kamenica2011bayesian}
E.~Kamenica and M.~Gentzkow.
\newblock Bayesian persuasion.
\newblock \emph{American Economic Review}, 101\penalty0 (6):\penalty0
  2590--2615, 2011.

\bibitem[Khantadze et~al.(2021)Khantadze, Kremer, and
  Skrzypacz]{khantadze2021persuasion}
D.~Khantadze, I.~Kremer, and A.~Skrzypacz.
\newblock Persuasion with multiple actions.
\newblock 2021.

\bibitem[Kleiner et~al.(2021)Kleiner, Moldovanu, and
  Strack]{kleiner2021extreme}
A.~Kleiner, B.~Moldovanu, and P.~Strack.
\newblock Extreme points and majorization: Economic applications.
\newblock \emph{Econometrica}, 89\penalty0 (4):\penalty0 1557--1593, 2021.

\bibitem[Kolesnikov et~al.(2022)Kolesnikov, Sandomirskiy, Tsyvinski, and
  Zimin]{kolesnikov2022beckmann}
A.~Kolesnikov, F.~Sandomirskiy, A.~Tsyvinski, and A.~P. Zimin.
\newblock Beckmann's approach to multi-item multi-bidder auctions.
\newblock \emph{arXiv preprint arXiv:2203.06837}, 2022.

\bibitem[Kolotilin(2018)]{kolotilin2018optimal}
A.~Kolotilin.
\newblock Optimal information disclosure: A linear programming approach.
\newblock \emph{Theoretical Economics}, 13\penalty0 (2):\penalty0 607--635,
  2018.

\bibitem[Kolotilin et~al.(2022)Kolotilin, Corrao, and
  Wolitzky]{kolotilin2022persuasion}
A.~Kolotilin, R.~Corrao, and A.~Wolitzky.
\newblock Persuasion with non-linear preferences.
\newblock \emph{UNSW Economics Working Paper 2022-03}, 2022.

\bibitem[Kravchenko(2024)]{kravchenko2024}
E.~Kravchenko.
\newblock Coherent distributions: {H}ilbert space approach and duality.
\newblock \emph{arXiv preprint arXiv:2405.04375}, 2024.

\bibitem[Laclau and Renou(2017)]{laclau2017public}
M.~Laclau and L.~Renou.
\newblock Public persuasion.
\newblock \emph{Manuscript}, 2017.

\bibitem[Lang(2022)]{lang2022feasible}
X.~Lang.
\newblock Feasible joint posterior beliefs with many states.
\newblock \emph{Available at SSRN 4077632}, 2022.

\bibitem[Levy et~al.(2022)Levy, Barreda, and Razin]{levy2022persuasion}
G.~Levy, I.~M.~d. Barreda, and R.~Razin.
\newblock Persuasion with correlation neglect: a full manipulation result.
\newblock \emph{American Economic Review: Insights}, 4\penalty0 (1):\penalty0
  123--138, 2022.

\bibitem[Lin and Liu(2022)]{lin2022credible}
X.~Lin and C.~Liu.
\newblock Credible persuasion.
\newblock In \emph{Proceedings of the 23rd ACM Conference on Economics and
  Computation}, pages 469--469, 2022.

\bibitem[Malamud and Schrimpf(2021)]{malamud2021persuasion}
S.~Malamud and A.~Schrimpf.
\newblock Persuasion by dimension reduction.
\newblock \emph{arXiv preprint arXiv:2110.08884}, 2021.

\bibitem[Mathevet et~al.(2020)Mathevet, Perego, and
  Taneva]{mathevet2020information}
L.~Mathevet, J.~Perego, and I.~Taneva.
\newblock On information design in games.
\newblock \emph{Journal of Political Economy}, 128\penalty0 (4):\penalty0
  1370--1404, 2020.

\bibitem[McKelvey and Page(1986)]{mckelvey1986common}
R.~D. McKelvey and T.~Page.
\newblock Common knowledge, consensus, and aggregate information.
\newblock \emph{Econometrica: Journal of the Econometric Society}, pages
  109--127, 1986.

\bibitem[Mertens et~al.(2015)Mertens, Sorin, and Zamir]{mertens2015repeated}
J.-F. Mertens, S.~Sorin, and S.~Zamir.
\newblock \emph{Repeated games}, volume~55.
\newblock Cambridge University Press, 2015.

\bibitem[Morris et~al.(2024)Morris, Oyama, and
  Takahashi]{morris2024implementation}
S.~Morris, D.~Oyama, and S.~Takahashi.
\newblock Implementation via information design in binary-action supermodular
  games.
\newblock \emph{Econometrica}, 92\penalty0 (3):\penalty0 775--813, 2024.

\bibitem[Morris(2020)]{morris2020notrade}
S.~E. Morris.
\newblock No trade and feasible joint posterior beliefs.
\newblock \emph{a working paper}, 2020.
\newblock URL
  \url{https://economics.mit.edu/sites/default/files/inline-files/no%20trade%206.pdf}.

\bibitem[Rachev and R{\"u}schendorf(1998)]{rachev1998mass}
S.~T. Rachev and L.~R{\"u}schendorf.
\newblock \emph{Mass Transportation Problems: Volume I: Theory}, volume~1.
\newblock Springer Science \& Business Media, 1998.

\bibitem[Rudin(1964)]{rudin1964principles}
W.~Rudin.
\newblock \emph{Principles of mathematical analysis}, volume~3.
\newblock McGraw-hill New York, 1964.

\bibitem[Santambrogio(2015)]{santambrogio2015optimal}
F.~Santambrogio.
\newblock Optimal transport for applied mathematicians.
\newblock \emph{Birk{\"a}user, NY}, 55\penalty0 (58-63):\penalty0 94, 2015.

\bibitem[Smolin and Yamashita(2022)]{smolin2022information}
A.~Smolin and T.~Yamashita.
\newblock Information design in concave games.
\newblock In \emph{Proceedings of the 23rd ACM Conference on Economics and
  Computation}, pages 870--870, 2022.

\bibitem[Steinerberger and Tsyvinski(2019)]{steinerberger2019tax}
S.~Steinerberger and A.~Tsyvinski.
\newblock Tax mechanisms and gradient flows.
\newblock Technical report, National Bureau of Economic Research, 2019.

\bibitem[Taneva(2019)]{taneva2019information}
I.~Taneva.
\newblock Information design.
\newblock \emph{American Economic Journal: Microeconomics}, 11\penalty0
  (4):\penalty0 151--85, 2019.

\bibitem[Villani(2009)]{villani2009optimal}
C.~Villani.
\newblock \emph{Optimal transport: old and new}, volume 338.
\newblock Springer, 2009.

\bibitem[Winkler(1988)]{winkler1988extreme}
G.~Winkler.
\newblock Extreme points of moment sets.
\newblock \emph{Mathematics of Operations Research}, 13\penalty0 (4):\penalty0
  581--587, 1988.

\bibitem[Yang and Zentefis(2024)]{yang2024monotone}
K.~H. Yang and A.~K. Zentefis.
\newblock Monotone function intervals: Theory and applications.
\newblock \emph{American Economic Review}, 114\penalty0 (8):\penalty0
  2239--2270, 2024.

\bibitem[Zhu(2022)]{zhu2022some}
T.~Zhu.
\newblock \emph{Some Problems on the Convex Geometry of Probability Measures}.
\newblock PhD thesis, UC Berkeley, 2022.

\bibitem[Ziegler(2020)]{ziegler2020adversarial}
G.~Ziegler.
\newblock Adversarial bilateral information design.
\newblock \emph{Working Paper}, 2020.

\end{thebibliography}

\appendix

\ed{
\section{Convexity and Closedness of the Set of  Feasible Conditional Distributions}\label{sec_convex_closed}
}

\ed{
\begin{lemma}\label{lm_closed_compact}
   The set of all feasible conditional distributions  $(\mu^\omega)_{\omega\in \Omega}$ is a convex subset of $\left(\Delta\left(\Delta(\Omega)^N\right)\right)^\Omega  $ closed in the topology of weak convergence.
\end{lemma}
}
\begin{proof}
 \ed{By Theorem~\ref{th_feasibility}, $(\mu^\omega)_{\omega\in \Omega}$ are feasible conditional distributions if and only if the marginals $(\mu_i^\omega)_{\omega\in \Omega}$ are feasible for each receiver~$i$. According to Observation~\ref{obs_one_receiver_feasibility}, $(\mu_i^\omega)_{\omega\in \Omega}$ are feasible if there is $\lambda_i$ with mean $p$ such that $\dd \mu_i^\omega=\frac{x(\omega)}{p(\omega)}\dd\lambda_i$.} Multiplying this identity by $p(\omega)$ and summing over all states, we conclude that 
 $\lambda_i$ can be expressed as $\sum_{\omega'} p(\omega')\mu_i^{\omega'}$. We conclude that
 $(\mu_i^\omega)_{\omega\in \Omega}$ is feasible if and only if 
$$p(\omega)\, \dd\mu_i^\omega=x(\omega)\ \dd\left(\sum_{\omega'\in\Omega}p(\omega')\mu_i^{\omega'}\right)$$
or, equivalently, in the integrated form:
	\begin{equation}\label{eq_integrated_admissibility}
	p(\omega)\cdot \int_{\Delta(\Omega)\times\ldots\times \Delta(\Omega)} \psi(x_i)\dd \mu^\omega(x_1,\ldots,x_n)-\int_{\Delta(\Omega)}x_i(\omega)\cdot \psi(x_i)\left(\sum_{\omega'\in\Omega}p(\omega') \dd\mu^{\omega'}(x_1,\ldots,x_n)\right)=0
	\end{equation}
	for all continuous functions $\psi:\, \Delta(\Omega)\to \R$.  Since this condition is linear in $(\mu^\omega)_{\omega\in\Omega}$, a convex combination of feasible distributions is also feasible. Since the integrands are continuous functions, the weak limit of a sequence of distributions satisfying the conditions also satisfies them. We get closedness.
\end{proof}

\ed{\section{The Existence of an Optimal Information Structure}\label{app_existence_primal}}
\ed{
\begin{lemma}
In a first-order persuasion problem $B=(\Omega,p,N,v)$ with upper semicontinuous~$v^\omega$ in each state $\omega$, the senders's objective $\E_I[v^\omega(x_1,\ldots, x_n)]$ attains its maximum at some information structure $I$.
\end{lemma}
}
\begin{proof}
\ed{Maximization over information structures can be replaced with maximization over feasible conditional belief distributions. The sender's problem becomes to maximize
\begin{equation}\label{eq_value_as_maximization_over_joint_distr}
    \sum_{\omega\in \Omega} \ p(\omega)\cdot \int_{\Delta(\Omega)\times\ldots \times\Delta(\Omega)} v^\omega\,\dd\mu^\omega
\end{equation}
over feasible conditional distributions $(\mu^\omega)_{\omega\in \Omega}$. Since every such feasible $(\mu^\omega)_{\omega\in \Omega}$ is induced by some $I$, it is enough to show the existence of the optimal $(\mu^\omega)_{\omega\in \Omega}$.}

The integral of an upper semicontinuous function over a compact set is an upper semicontinuous function of the distribution in the weak topology \cite*[Lemma~4.3]{villani2009optimal}. Hence, the objective  in \eqref{eq_value_as_maximization_over_joint_distr}
is upper semicontinuous. An upper semicontinuous function on a compact set attains its maximum. 
By Lemma~\ref{lm_closed_compact}, the set of all feasible conditional belief distributions
is a closed subset of $\big(\Delta\big(\Delta(\Omega)^N\big)\big)^\Omega$ and, hence, is compact since the set of all probability distributions over a compact set is compact in the weak topology. We conclude that the maximum in \eqref{eq_value_as_maximization_over_joint_distr}  is attained. Thus both the optimal information structure and the optimal conditional distributions exist.

\end{proof}

\section{Proof of Proposition~\ref{prop_dual value representation}}\label{app_proof_of_duality}

To prove the dual representation for the value~\eqref{eq:dual_value_explicit} of the persuasion problem $B$, we introduce an auxiliary zero-sum game $G$ such that the $\max\inf$-value of $G$ coincides with the value of $B$ and then exchange $\max$ and $\inf$ via Sion's minimax theorem.

By Corollary~\ref{cor_persuasion_as_transport}, to get the value of $B$, it is enough to maximize 
	\begin{equation}\label{eq_maximizers_objective}
	\sum_{\omega\in \Omega} p(\omega)\cdot\int_{\Delta(\Omega)\times\ldots\times\Delta(\Omega)} v^\omega(x_1,\ldots x_n)\dd\mu^\omega(x_1,\ldots,x_n)
	\end{equation}
	over a family of measures $(\mu^\omega)_{\omega\in \Omega}\in \Big(\Delta\big(\Delta(\Omega)^N\big)\Big)^\Omega$ with feasible one-agent marginals. The feasibility of marginals requires that the Radon-Nikodym derivatives of the marginals $\mu_i^\omega$ of $\mu^\omega$ with respect to some $\lambda_i\in \Delta(\Delta(\Omega))$ satisfy $\frac{\dd \mu_i^\omega}{\dd \lambda_i}(x_i)=\frac{x_i(\omega)}{p(\omega)}$.
	As in the proof of Lemma~\ref{lm_closed_compact}, from this equation, we conclude that	$\lambda_i=\sum_{\omega'\in\Omega}p(\omega')\cdot \mu_i^{\omega'}$ and, hence, the feasibility of the marginals is equivalent to the identity
	$$0=p(\omega)\,\dd\mu_i^\omega(x_i)-x_i(\omega)\cdot\sum_{\omega'\in\Omega}p(\omega')\,\dd\mu_i^{\omega'}(x_i),$$
	which can be rewritten in the integrated form as follows:
	\begin{align}\label{eq_integrated_admissibility2}
	0&=\int_{\Delta(\Omega)\times\ldots\times \Delta(\Omega)} \psi_i^\omega(x_i)\dd \mu^\omega(x_1,\ldots,x_n)\\
 \notag
 &-\int_{\Delta(\Omega)\times\ldots\times \Delta(\Omega)}x_i(\omega)\cdot \psi_i^\omega(x_i)\left(\sum_{\omega'\in\Omega}p(\omega') \dd\mu^{\omega'}(x_1,\ldots,x_n)\right)
	\end{align}
	for all continuous functions $\psi_i^\omega:\, \Delta(\Omega)\to \R$.

	We now define the game $G$. In this game, the maximizer aims to maximize~\eqref{eq_maximizers_objective}, and we allow her to pick an arbitrary collection of probability measures $(\mu^\omega)_{\omega\in \Omega}$, which may have non-feasible marginals.
	However, the minimizer can penalize her for violation of the identity \eqref{eq_integrated_admissibility2} by selecting a family of continuous functions $(\psi_i^\omega)_{i\in N,\omega\in \Omega}$.
	The payoff function is defined as follows
	\begin{align*}
		G\Big[\big(\mu^\omega\big)_{\omega\in \Omega},\big(\psi_i^\omega\big)_{i\in N,\omega\in \Omega}\Big]&= 
		\sum_{\omega\in \Omega}\left( p(\omega)\cdot\int_{\Delta(\Omega)\times\ldots\times\Delta(\Omega)} v^\omega(x_1,\ldots x_n)\dd\mu^\omega(x_1,\ldots,x_n)\right.\\
		&- \sum_{i\in N}\left(p(\omega)\cdot \int_{\Delta(\Omega)\times\ldots\times \Delta(\Omega)} \psi_i^\omega(x_i)\,\dd  \mu^\omega(x_1,\ldots,x_n)\right.\\
		&\left.-\int_{\Delta(\Omega)\times\ldots\times\Delta(\Omega)}x_i(\omega)\cdot \psi_i^\omega(x_i)\left(\sum_{\omega'\in\Omega}p(\omega') \dd\mu^{\omega'}(x_1,\ldots,x_n)\right)\right)	
	\end{align*}
	If the maximizer selects $(\mu^\omega)_{\omega\in \Omega}$ with feasible marginals, then each of the $N$ terms in the sum over $i\in N$ equals zero. On the other hand, if the feasibility constraint on marginals is violated, the minimizer can arbitrarily lower the payoff by choosing $(\psi_i^\omega)_{i\in N,\omega\in \Omega}$. Therefore, $$\val[B]=\sup_{(\mu^\omega)_{\omega\in \Omega}}\inf_{(\psi_i^\omega)_{i\in N,\omega\in \Omega}}G\Big[\big(\mu^\omega\big)_{\omega\in \Omega},\big(\psi_i^\omega\big)_{i\in N,\omega\in \Omega}\Big].$$

The assumptions of Sion's minimax theorem\footnote{Sion's theorem claims that $\sup_{x\in X} \inf_{y\in Y} G(x,y)=\inf_{y\in Y} \sup_{x\in X}  G(x,y)$  if $X$ and $Y$ are convex subsets of linear topological spaces, at least one of them is compact, and $G$ is an upper semicontinuous quasiconcave function of the first argument and lower semicontinuous quasiconvex of the second. See~\cite*{mertens2015repeated}, Chapter~I.1.}	 are satisfied by $G\Big[\big(\mu^\omega\big)_{\omega\in \Omega},\big(\psi_i^\omega\big)_{i\in N,\omega\in \Omega}\Big]$ and we can exchange $\sup_{(\mu^\omega)_{\omega\in \Omega}}$ and $\inf_{(\psi_i^\omega)_{i\in N,\omega\in \Omega}}$. Indeed, the set of probability measures on a compact metric space is compact in the weak topology, $G$ is an affine function of strategies of each of the players (and thus both convex and concave), it is an upper semicontinuous function of $(\mu^\omega)_{\omega\in \Omega}$ in the weak topology (see Lemma~4.3 in Section~4 of \cite*{villani2009optimal}) and a continuous function of $(\psi_i^\omega)_{i\in N,\omega\in \Omega}$ in the topology induced by the $\sup$-norm on continuous functions. 
	We obtain 
	$$\val[B]=\inf_{(\psi_i^\omega)_{i\in N,\omega\in \Omega}}\sup_{(\mu^\omega)_{\omega\in \Omega}}G\Big[\big(\mu^\omega\big)_{\omega\in \Omega},\big(\psi_i^\omega\big)_{i\in N,\omega\in \Omega}\Big].$$ 
	
	For a compact metric space $X$, we have $\max_{\nu\in\Delta(X)}\int h(x)d\nu(x)=\max_{x\in X}h(x)$ for any upper semicontinuous function $h$ on $X$; in particular, both maxima are attained. Hence 
	the internal unconstrained maximization over $(\mu^\omega)_{\omega\in \Omega}$  leads to the pointwise maxima of the corresponding integrands, and we get
			\begin{align*}
			\val[B]=\inf_{(\psi_i^\omega)_{i\in N,\omega\in \Omega}} \sum_{\omega\in\Omega} p(\omega)\cdot \max_{(x_i)_{i\in N}\in \Delta(\Omega)^N}
			\left(v^\omega(x_1,\ldots,x_n)- \sum_{i\in N}\left(\psi_i^\omega(x_i)-\sum_{\omega'\in\Omega} x_i(\omega')\cdot \psi_i^{\omega'}(x_i) \right)\right). 
		\end{align*}
	For a family of functions $(\psi_i^\omega)_{i\in N,\omega\in \Omega}$ define a new family $(\varphi_i^\omega)_{i\in N,\omega\in \Omega}$ by the formula.
	\begin{equation}\label{eq_psi_to_phi}
	\varphi_i^\omega(x_i)=\psi_i^\omega(x_i)-\sum_{\omega'\in\Omega} x_i(\omega')\cdot \psi_i^{\omega'}(x_i),\qquad x_i\in\Delta(\Omega).
	\end{equation}
 The new family satisfies an additional condition
 \begin{equation}\label{eq_budget_balance}
 \sum_{\omega\in\Omega} x_i(\omega)\varphi_i^\omega(x_i)= 0,\qquad x_i\in\Delta(\Omega),
 \end{equation}
 and gives the same value to the objective as the original one. We obtain the following:
	\begin{align}\label{eq:dual_value_explicit_theta}
			\val[B]=&\inf_{{\footnotesize \begin{array}{c} 
			\mbox{continuous}\\
			\mbox{$(\varphi_i^\omega)_{i\in N,\omega\in \Omega}$ such that}\\		\mbox{$\sum_{\omega\in\Omega}x_i(\omega) \varphi_i^\omega(x_i)\equiv 0$}
				\end{array}}} \sum_{\omega\in\Omega} p(\omega)\cdot \max_{(x_i)_{i\in N}\in \Delta(\Omega)^N}
			\left(v^\omega(x_1,\ldots,x_n)-  \sum_{i\in N} \varphi_i^\omega(x_i)\right).
		\end{align}
	Finally, 
	 we pick arbitrary
	$V^\omega\geq \max_{(x_i)_{i\in N}\in \Delta(\Omega)^N}
			\left(v^\omega(x_1,\ldots,x_n)-  \sum_{i\in N} \varphi_i^\omega(x_i)\right)$
	and obtain
	\begin{equation}\label{eq_duality_duplicated_in_appendix}
	\val[B]=\inf_{{\footnotesize \begin{array}{c}  \mbox{{$V^\omega\in\R$, continuous  $\varphi_i^\omega$ on $\Delta(\Omega)$ such that}}\\
			\mbox{$v^\omega(x_1,\ldots,x_n)\leq V^\omega+\sum_{i\in N}\varphi_i^\omega(x_i)$}\\
					\mbox{and $\sum_{\omega\in\Omega}x_i(\omega) \varphi_i^\omega(x_i)=0$} 	\end{array}}}\sum_{\omega\in\Omega} p(\omega)V^\omega.
	\end{equation}	
	which coincides with the desired formula from the statement of Proposition~\ref{prop_dual value representation}.

\medskip

\paragraph{The existence of optima.}
 Here we demonstrate that for continuous utility functions $v^\omega$ the infimum in~\eqref{eq_duality_duplicated_in_appendix} is attained, i.e., optimal $(V^\omega,\,\varphi_i^\omega)_{i\in N,\omega\in \Omega}$ exist.

The idea is to show that we can restrict the minimization to some compact set and then extract a subsequence converging to an optimum.
The restrictions that we can impose on $(V^\omega)_{\omega\in \Omega}$ and $\varphi_i^\omega$ are presented in the following lemma. To formulate it, 
we define the norm of the utility function by
$$\|v\|_\infty=\max_{\omega\in\Omega,\, x_1,\ldots,x_n\in\Delta(\Omega)} \big|v^\omega(x_1,\ldots,x_n)\big|$$ 
and its modulus of continuity, by
$$D_v(\varepsilon)=\hskip -1cm \max_{\footnotesize{\begin{array}{c}\omega\in\Omega,i\in N\\ 
x_1,\ldots,x_{i-1},x_{i+1},\ldots,x_n\in\Delta(\Omega)\\
x,x'\in\Delta(\Omega)\,:\,
|x-x'|\leq\varepsilon
\end{array}}}\hskip -1cm\Big|v^\omega\big(x_1,\ldots,x_{i-1},x,x_{i+1},\ldots,x_n\big)-v^\omega\big(x_1,\ldots,x_{i-1},x',x_{i+1},\ldots,x_n\big)\Big|,$$
where $|x-x'|=\sum_{\omega\in\Omega} |x(\omega)-x(\omega')|$ is the total variation distance between $x$ and $x'\in\Delta(\Omega)$.
\begin{lemma}\label{lm_regularity_estimates}
Restricting the minimization in~\eqref{eq_duality_duplicated_in_appendix} to $(V^\omega,\,\varphi_i^\omega)_{i\in N,\omega\in\Omega}$ such that
\begin{align}\label{eq_bound_V_omega}
-\|v\|_\infty  &\leq \Big|V^\omega\Big|\leq  \frac{2-p(\omega)}{p(\omega)}\cdot \|v\|_\infty,\\
\label{eq_upper_lower_phi}
	 	 -\frac{2n}{p(\omega)}\cdot \|v\|_\infty&\leq\Big|{{\varphi}}_i^\omega(x)\Big|\leq \frac{2}{p(\omega)}\cdot \|v\|_\infty,\qquad x\in\Delta(\Omega),\\
\label{eq_modulus_continuity_phi}
	  \Big|{{\varphi}}_i^\omega(x)-{{\varphi}}_i^\omega(x')\Big|&\leq 2\cdot D_v\Big(|x-x'|\Big)+\frac{2n}{\min_{\omega'\in\Omega}p(\omega')}\cdot\|v\|_\infty\cdot |x-x'|,\qquad x,\,x'\in\Delta(\Omega),
	\end{align}
does not affect the optimal value.
\end{lemma}

We first check that this lemma implies the existence of the optimal $(V^\omega,\,\varphi_i^\omega)_{i\in N,\omega\in \Omega}$ and then prove the lemma.
Consider a sequence $(V^{\omega,t},\,\varphi_i^{\omega,t})_{i\in N,\omega\in \Omega}$ indexed by a parameter $t=1,2,\ldots$ and such that the objective in~\eqref{eq_duality_duplicated_in_appendix} converges to its optimum along this sequence, as $t$ goes to infinity.
By Lemma~\ref{lm_regularity_estimates}, we can additionally require that $(V^{\omega,t},\,\varphi_i^{\omega,t})_{i\in N,\omega\in \Omega}$ satisfy conditions~\eqref{eq_bound_V_omega}, \eqref{eq_upper_lower_phi} and~\eqref{eq_modulus_continuity_phi} for each $t$.
The set of numbers defined by~\eqref{eq_bound_V_omega} is compact as a closed bounded subset of~$\R^\Omega$. 
Functions satisfying \eqref{eq_upper_lower_phi} and~\eqref{eq_modulus_continuity_phi} are uniformly bounded and uniformly continuous and thus, by the
	Arzel\`{a}–Ascoli theorem (see~\cite*{rudin1964principles}), this class of functions compact in the topology of 
the space of continuous functions (induced by the $\sup$-norm). The product of compact sets is compact and, hence, the sequence $(V^{\omega,t},\,\varphi_i^{\omega,t})_{i\in N,\omega\in \Omega}$ belongs to a compact set. Let us extract a converging subsequence and denote its limit by $(V^\omega,\,\varphi_i^\omega)_{i\in N,\omega\in \Omega}$. The objective in~\eqref{eq_duality_duplicated_in_appendix} is continuous, and the constraints are closed. Hence, the collection $(V^\omega,\,\varphi_i^\omega)_{i\in N,\omega\in \Omega}$ gives the optimal value to the objective, satisfies the constraints, and thus is optimal.
	
	\medskip
To complete the proof of Proposition~\ref{prop_dual value representation}, it remains to prove the lemma.
\begin{proof}[Proof of Lemma~\ref{lm_regularity_estimates}.]
	For a given family  $(\varphi_i^\omega)_{i\in N,\omega\in \Omega}$ of continuous functions satisfying~\eqref{eq_budget_balance}, let $V^\omega\big[(\varphi_i^\omega)_{i\in N,\omega\in \Omega}\big]$ be the minimal value of $V^\omega$ such that $(V^\omega,\, \varphi_i^\omega)$ satisfy the constraints of~\eqref{eq_duality_duplicated_in_appendix}: 
	\begin{equation}\label{eq_M_omega_def}
V^\omega\big[(\varphi_i^\omega)_{i\in N,\omega\in \Omega}\big]=\max_{(x_i)_{i\in N}\in \Delta(\Omega^N)}
			\left(v^\omega(x_1,\ldots,x_n)- \sum_{i\in N}\varphi_i^\omega(x_i)\right).
	\end{equation}
	Without loss of generality, we can assume that $V^\omega$ in~\eqref{eq_duality_duplicated_in_appendix} is given by $V^\omega\big[(\varphi_i^\omega)_{i\in N,\omega\in \Omega}\big]$ and, hence, $V^\omega$ is determined by functions $(\varphi_i^\omega)_{i\in N,\omega\in \Omega}$, which remain the only free parameter in the minimization. In particular, to prove the bounds \eqref{eq_bound_V_omega} on $V^\omega$ it is enough to show that we can restrict minimization to $(\varphi_i^\omega)_{i\in N,\omega\in \Omega}$ such that
	\begin{equation}\label{eq_V_omega_optimal}
	    -\|v\|_\infty  \leq \Big|V^\omega\big[(\varphi_i^\omega)_{i\in N,\omega\in \Omega}\big]\Big|\leq  \frac{2-p(\omega)}{p(\omega)}\cdot \|v\|_\infty.
	\end{equation}
	Recall that $\delta_\omega\in \Delta(\Omega)$ is the point mass at the state $\omega$. Plugging $x_i=\delta_\omega$ for each $i$ into \eqref{eq_M_omega_def}, we obtain the following lower bound:
	$$V^\omega\big[(\varphi_i^\omega)_{i\in N,\omega\in \Omega}\big]\geq v^\omega\big(\delta_\omega,\ldots,\delta_\omega\big)\geq -\|v\|_\infty.$$
	Hence, the lower bound in~\eqref{eq_V_omega_optimal} holds.
	
	The optimal value of~\eqref{eq_duality_duplicated_in_appendix} cannot exceed the best value of the objective attained at the zero functions $\varphi_i^\omega$. Hence, minimization can be restricted to
	$(\varphi_i^\omega)_{i\in N,\omega\in\Omega}$
 such that 
	 $$\sum_{\omega\in \Omega} p(\omega)\cdot V^\omega\big[(\varphi_i^\omega)_{i\in N,\omega\in \Omega}\big]\leq \sum_{\omega\in \Omega} p(\omega)\cdot V^\omega\big[(0)_{i\in N,\omega\in \Omega}\big].$$
	 Since the right-hand side does not exceed $\|v\|_\infty$, we get 
	 \begin{equation}\label{eq_upper_bound_on_objective}
	  \sum_{\omega\in \Omega} p(\omega)\cdot V^\omega\big[(\varphi_i^\omega)_{i\in N,\omega\in \Omega}\big]\leq  \|v\|_\infty.
	 \end{equation}
	Changing all summands on the left-hand side of~\eqref{eq_upper_bound_on_objective} except one to their lower bounds and transferring them to the right-hand side, we get 
	\begin{equation}\label{eq_upper_V_omega_opt}
	V^\omega\big[(\varphi_i^\omega)_{i\in N,\omega\in \Omega}\big]\leq \frac{2-p(\omega)}{p(\omega)}\cdot \|v\|_\infty.
	\end{equation}
	We obtain the upper bound in \eqref{eq_V_omega_optimal}. Moreover, 
	this inequality implies an upper bound on $\varphi_i^\omega$. Indeed, let us plug $x_j=\delta_\omega$ for all receivers $j$ except $j=i$ into the objective of~\eqref{eq_M_omega_def}. The value of the objective on this input cannot exceed the optimal value and, taking into account that $\varphi_j^\omega(\delta_\omega)=0$ thanks to~\eqref{eq_budget_balance}, we deduce 
	$$v^\omega(\delta_\omega,\ldots,\delta_\omega,x_i,\delta_\omega,\ldots,\delta_\omega)+\varphi_i^\omega(x_i)\leq V^\omega\big[(\varphi_i^\omega)_{i\in N,\omega\in \Omega}\big].$$
	Consequently,
	\begin{equation}\label{eq_upper_bound_varphi}
	\varphi_i^\omega(x)\leq \frac{2}{p(\omega)}\cdot\|v\|_\infty,
	\end{equation}
	i.e, the upper bound in~\eqref{eq_upper_lower_phi} holds.
	\medskip

	To summarize: without loss of generality, the minimization in~\eqref{eq_duality_duplicated_in_appendix} can be restricted to families of continuous functions $(\varphi_i^\omega)_{i\in N,\omega\in \Omega}$  satisfying~\eqref{eq_budget_balance} and~\eqref{eq_upper_bound_on_objective}; the upper bound~\eqref{eq_upper_bound_varphi} is satisfied for all such families automatically, as well as the bounds~\eqref{eq_V_omega_optimal}.
	Now, we consider such a family, fix a receiver $k\in N$ and show that we can replace the functions $(\varphi_k^\omega)_{\omega\in\Omega}$ by $(\widetilde{\widetilde{\varphi}}_k^\omega)_{\omega\in\Omega}$ keeping the rest of the family unchanged so that the new family satisfies the same requirements, the value of the objective remains the same or improves, and most importantly, the functions $(\widetilde{\widetilde{\varphi}}_k^\omega)_{\omega\in\Omega}$ additionally satisfy bounds \eqref{eq_upper_lower_phi} and \eqref{eq_modulus_continuity_phi}.
	Define $\widetilde{\varphi}_k^\omega$ by
	$$\widetilde{\varphi}_k^\omega(x)=\max_{x_1,\ldots, x_{k-1},x_{k+1},\ldots,x_n} \left(v^\omega(x_1,\ldots, x_{k-1},x,x_{k+1},\ldots,x_n)-\sum_{i\in N\setminus\{k\}}\varphi_i^\omega(x_i)\right)-V^\omega\big[(\varphi_i^\omega)_{i\in N,\omega\in \Omega}\big].$$
	From the definition, we see that 
	$$V^\omega\Big[\Big((\widetilde{\varphi}_k^\omega)_{\omega\in\Omega},(\varphi_i^\omega)_{i\in N\setminus\{k\},\omega\in \Omega}\Big)\Big]=V^\omega\big[(\varphi_i^\omega)_{i\in N,\omega\in \Omega}\big]$$
	and, moreover, the functions $\widetilde{\varphi}_k^\omega$ are pointwise minimal among all the functions with this property. Hence, ${\varphi}_k^\omega\geq\widetilde{\varphi}_k^\omega$.

	The functions $(\widetilde{\varphi}_k^\omega)_{\omega\in\Omega}$ may violate the requirement~\eqref{eq_budget_balance}. To enforce this requirement, we set
$$\widetilde{\widetilde{\varphi}}_k^\omega(x)={\widetilde{\varphi}}_k^\omega(x)-\sum_{\omega'\in\Omega}x(\omega')\cdot {\widetilde{\varphi}}_k^{\omega'}(x).$$
	The functions $\widetilde{\widetilde{\varphi}}_k^\omega$ satisfy~\eqref{eq_budget_balance}. Since ${\varphi}_k^\omega\geq\widetilde{\varphi}_k^\omega$,
	$$\sum_{\omega'\in\Omega}x(\omega')\cdot {\widetilde{\varphi}}_k^\omega(x)\leq \sum_{\omega'\in\Omega}x(\omega')\cdot {{\varphi}}_k^\omega(x)=0,$$
	and we see that 
	$\widetilde{\widetilde{\varphi}}_k^\omega\geq {\widetilde{\varphi}}_k^\omega$. Therefore, 
	$$V^\omega\Big[\Big((\widetilde{\widetilde{\varphi}}_k^\omega)_{\omega\in\Omega},(\varphi_i^\omega)_{i\in N\setminus\{k\},\omega\in \Omega}\Big)\Big]\leq
	V^\omega\Big[\Big((\widetilde{\varphi}_k^\omega)_{\omega\in\Omega},(\varphi_i^\omega)_{i\in N\setminus\{k\},\omega\in \Omega}\Big)\Big],$$
	and so	replacing ${{\varphi}}_k^\omega$ by $\widetilde{\widetilde{\varphi}}_k^\omega$ can only improve the objective in~\eqref{eq:dual_value_explicit_theta}.
	
	We conclude that the constructed family satisfies the conditions~\eqref{eq_budget_balance} and~\eqref{eq_upper_bound_on_objective} (hence, the upper bound~\eqref{eq_upper_bound_varphi}  also holds) and the value of the objective remains the same or improves.
	Now let us check that  $\widetilde{\widetilde{\varphi}}_k^\omega$ satisfies the lower bound in \eqref{eq_upper_lower_phi} and the bound \eqref{eq_modulus_continuity_phi}. 
	
	 From the definition of ${\widetilde{\varphi}}_k^\omega$ the bounds \eqref{eq_upper_V_omega_opt} and~\eqref{eq_upper_bound_varphi}, we obtain
	$$ -\frac{2n}{p(\omega)}\cdot \|v\|_\infty\leq {\widetilde{\varphi}}_k^\omega(x).$$
	Since $\widetilde{\widetilde{\varphi}}_k^\omega\geq {\widetilde{\varphi}}_k^\omega$, the same lower bound holds for $\widetilde{\widetilde{\varphi}}_k^\omega$. Thus, $\widetilde{\widetilde{\varphi}}_k^\omega$ satisfies both bounds of \eqref{eq_upper_lower_phi}.

	To prove~\eqref{eq_modulus_continuity_phi}, we estimate the difference $\Big|{\widetilde{\varphi}}_k^\omega(x)-{\widetilde{\varphi}}_k^\omega(x')\Big|$ first. By the definition of $D_v(\varepsilon)$,
	$$v^\omega\big(x_1,\ldots,x_{i-1},x,x_{i+1},\ldots,x_n\big)+D_v\Big(|x-x'|\Big)\geq v^\omega\big(x_1,\ldots,x_{i-1},x',x_{i+1},\ldots,x_n\big)$$
	for any $x, x'\in\Delta(\Omega)$ and all $x_1,\ldots, x_{k-1},x_{k+1},\ldots, x_n\in \Delta(\Omega)$. Subtracting $\sum_{i\in N\setminus\{k\}}\varphi_i^\omega(x_i)+V^\omega\big[(\varphi_i^\omega)_{i\in N,\omega\in \Omega}\big]$ from both sides and taking maximum over $x_1,\ldots, x_{k-1},x_{k+1},\ldots, x_n\in \Delta(\Omega)$, we get
	$${\widetilde{\varphi}}_k^\omega(x)+ D_v\Big(|x-x'|\Big) \geq {\widetilde{\varphi}}_k^\omega(x').$$
	Combining this inequality with the one where the roles of $x$ and $x'$ are exchanged, we obtain 
	\begin{equation}\label{eq_upper_phi_tilde}
	\Big|{\widetilde{\varphi}}_k^\omega(x)-{\widetilde{\varphi}}_k^\omega(x')\Big|\leq  D_v\Big(|x-x'|\Big).
	\end{equation}
	From the definition of $\widetilde{\widetilde{\varphi}}_k^\omega$,
	$$\widetilde{\widetilde{\varphi}}_k^\omega(x)-\widetilde{\widetilde{\varphi}}_k^\omega(x')=\Big( {\widetilde{\varphi}}_k^\omega(x)-{\widetilde{\varphi}}_k^\omega(x')\Big)-\sum_{\omega'\in\Omega} x(\omega')\Big(\widetilde{\varphi}_k^{\omega'}(x)-\widetilde{\varphi}_k^{\omega'}(x')\Big)-\sum_{\omega'\in\Omega} \Big(x(\omega')-x'(\omega')\Big)\widetilde{\varphi}_k^{\omega'}(x').$$
	Estimating the first two terms on the right-hand side using~\eqref{eq_upper_phi_tilde} and bounding the absolute value of the last term by $|x-x'|\cdot \max_{\omega',x}|\widetilde{\varphi}_k^{\omega'}(x)|$, we see that $\widetilde{\widetilde{\varphi}}_k^\omega$ satisfies~\eqref{eq_modulus_continuity_phi}.
	
	\medskip
	Sequentially replacing $\varphi_k^\omega$  in $(\varphi_i^\omega)_{i\in N,\omega\in \Omega}$ by $\widetilde{\widetilde{\varphi}}_k^\omega$ for all receivers $k\in N$, we obtain a collection of functions that satisfies \eqref{eq_upper_lower_phi} and \eqref{eq_modulus_continuity_phi}, while the value of the objective in~\eqref{eq_duality_duplicated_in_appendix} remains the same or improves. Thus restricting the minimization in~\eqref{eq_duality_duplicated_in_appendix} to families that satisfy \eqref{eq_upper_lower_phi} and \eqref{eq_modulus_continuity_phi} does not affect the optimal value.
	\end{proof}

\section{Proofs for one-state persuasion}\label{app_one_state}

\begin{proof}[Proof of Lemma~\ref{lm_projection_admissible}]
Let us demonstrate the necessity of the condition~\eqref{eq_admissible_marginals_omega0}. In other words, we need to show that if $(\lambda^\omega)_{\omega\in\Omega}$ are feasible in one-agent problem, then
	\begin{align*}
		\int_{\Delta(\Omega)} \frac{x(\omega)}{x(\omega_0)}\dd\lambda^{\omega_0}(x)   &\leq \frac{p(\omega)}{p(\omega_0)},\qquad \omega\in\Omega\setminus\{\omega_0\}. 
	\end{align*} 
By Observation~\ref{obs_one_receiver_feasibility}, there exists $\lambda$ such that the Radon-Nikodym derivative
\begin{equation}\label{eq_one_state_admissible}
\frac{\dd\lambda^{\omega}}{\dd\lambda}(x)=\frac{x({\omega})}{p(\omega)}
\end{equation}
for all $\omega$. Let $\varepsilon>0$ be the small parameter. Hence, $\frac{\dd\lambda^{\omega_0}}{\dd\lambda}(x)\leq\frac{\max\{x({\omega_0}),\varepsilon\}}{p({\omega_0})}$ or, equivalently,
\begin{equation}\label{eq_upper_bound_one_state_espilon}
\frac{1}{\max\{x({\omega_0}),\varepsilon\}}\dd\lambda^{\omega_0}(x)\leq \frac{1}{p({\omega_0})}\dd\lambda(x).
\end{equation}
By \eqref{eq_one_state_admissible}, 
$\frac{x(\omega)}{p(\omega)}\dd\lambda(x)=\dd\lambda^\omega(x)$. Applying this identity to~\eqref{eq_upper_bound_one_state_espilon}, we get
$$\frac{x(\omega)}{\max\{x({\omega_0}),\varepsilon\}}\dd\lambda^{\omega_0}(x)\leq \frac{p(\omega)}{p({\omega_0})}\dd\lambda^\omega(x).$$
Integrating this inequality over $\Delta(\Omega)$, we obtain
$$\int_{\Delta(\Omega)} \frac{x(\omega)}{\max\{x({\omega_0}),\varepsilon\}}\dd\lambda^{\omega_0}(x)\leq \frac{p(\omega)}{p({\omega_0})}.$$
Letting $\varepsilon$ go to zero gives~\eqref{eq_admissible_marginals_omega0}. 

Now we check the sufficiency. For given $\lambda^{\omega_0}$ satisfying~\eqref{eq_admissible_marginals_omega0} we need to construct $\lambda^\omega$ with $\omega\in \Omega\setminus\{\omega_0\}$ such that the collection $(\lambda^{\omega})_{\omega\in\Omega}$ is feasible.
The idea is to use formula~\eqref{eq_one_state_admissible} to define $\lambda$ first. Set 
$$\dd\widetilde{\lambda}(x)=\frac{p(\omega_0)}{x(\omega_0)}\dd\lambda^{\omega_0}(x).$$
The measure $\widetilde{\lambda}$ may not be a probability measure, and its mean may not equal $p$. To make a probability measure with the desired mean out of $\widetilde{\lambda}$, we  define $\lambda$ by
\begin{equation}\label{eq_unconditional_from_omega_0}
\lambda=\widetilde{\lambda}+\sum_{\omega\in\Omega\setminus\{\omega_0\}}\left(p(\omega)-\int_{\Delta(\Omega)} x(\omega)\dd\widetilde{\lambda}(x)\right)\cdot \delta_\omega,
\end{equation}
where $\delta_\omega$ denotes the point mass at $\omega$. By \eqref{eq_admissible_marginals_omega0}, the coefficients in~\eqref{eq_unconditional_from_omega_0} are non-negative and, hence, $\lambda$ is a non-negative measure. By the construction $\int x(\omega)\dd\lambda=p(\omega)$ and so the mean of $\lambda$ is $p$. Summing up these equalities, we see that $\lambda$ is a probability measure. For $\omega\ne \omega_0$, define $\lambda^\omega$ by~\eqref{eq_one_state_admissible}; $\lambda^\omega$ is a probability measure since the mean of $\lambda$ is $p$. To show that $(\lambda^\omega)_{ \omega\in\Omega}$ are feasible, it remains to check that the condition~\eqref{eq_one_state_admissible} is satisfied at $\omega_0$. Since $x(\omega_0)\dd\delta_\omega=0$ for $\omega\ne\omega_0$, we get $x(\omega_0)\dd\lambda(x)=x(\omega_0)\dd\widetilde{\lambda}(x)$. By the definition of $\widetilde{\lambda}$,
$$\dd\lambda^{\omega_0}(x)=\frac{x(\omega_0)}{p(\omega_0)}\dd\widetilde{\lambda}(x)=\frac{x(\omega_0)}{p(\omega_0)}\dd\lambda(x).$$
From this identity, we conclude that $\frac{\dd\lambda^{\omega_0}}{\dd\lambda}(x)=\frac{x(\omega_0)}{p(\omega_0)}$, which completes the proof of the Lemma~\ref{lm_projection_admissible}.
\end{proof}
\medskip

\begin{proof}[Proof of Lemma~\ref{lm_one_state_counting_extreme}:]
Combining Corollary~\ref{cor_persuasion_as_transport} and Lemma~\ref{lm_projection_admissible}, we obtain that the value of a one-state persuasion problem $B$ can be represented as follows:
\begin{equation*}
	\val[B]=p(\omega_0)\cdot \sup_{{\footnotesize \begin{array}{c} \pi\in \Delta\big(\Delta(\Omega)\times\ldots\times \Delta(\Omega)\big)\\ \mbox{such that the marginals satisfy \eqref{eq_admissible_marginals_omega0}}
	\end{array}}	
	}\int_{\Delta(\Omega)\times\ldots\times \Delta(\Omega)} v(x_1,\ldots,x_n)\,\dd\pi(x_1,\ldots,x_n).
\end{equation*}
Our goal is to check that, in this formula, it is enough to maximize over atomic $\pi$ with a certain bound on the number of atoms in the support.

Let $\F_n^{\omega_0}(p)$ be the set of all distributions $\pi$ satisfying the inequalities \eqref{eq_admissible_marginals_omega0}. Since these inequalities are linear, $\F_n^{\omega_0}(p)$ is a convex set. The objective linearly depends on $\pi\in \F_n^{\omega_0}(p)$ and, hence, by the Bauer principle, it is enough to restrict the maximization to the extreme points of the set $\F_n^{\omega_0}(p)$.

To describe the extreme points of $\F_n^{\omega_0}(p)$, we discuss how the set of extreme points changes when we intersect a convex set with half-spaces; see \cite*[Theorem~2.1]{winkler1988extreme} for details.
Let $X$ be a convex set with extreme points $X^*\subset X$ and any $H$ be a half-space. The set of extreme points of $X\cap H$ consists of the union of $X^*\cap H$ and extreme points of 
$(\partial H\cap X)^*$ that are convex combinations $\alpha x+(1-\alpha)x'$ of $x,x'\in X^*$ satisfying the condition $\alpha x+(1-\alpha)x'\in \partial H$, where $\partial H$ denotes the boundary of $H$.  
Similarly, for the intersection $X\cap \bigcap_{q=1}^Q H_q$ with the family of half-spaces, any extreme point $x^*$ is given by a convex combination of at most $k+1$ extreme points of $X$, where $k$ is the number of $H_q$ such that $x^*\in\partial H_q$.   

Applying this general statement to our case, we put $X=\Delta(\Delta(\Omega)\times\ldots\times\Delta(\Omega))$ and define the half-spaces $H_{i,\omega}$, $i\in N$, $\omega\in\Omega\setminus\{\omega_0\}$ as the set of signed measures satisfying inequalities~\eqref{eq_admissible_marginals_omega0} with given $i$ and $\omega$. Since the extreme points of $X$ are the point masses, we conclude that
any extreme point of $\F_n^{\omega_0}(p)$
is an atomic measure with at most $|N|(|\Omega|-1)+1$ atoms. Hence, one can restrict the maximization to such measures.

This statement can be strengthened. Let $n_i(\pi)$ be the number of ``active'' inequalities for the receiver $i$, i.e., $n_i(\pi)$ is the number of those inequalities from~\eqref{eq_admissible_marginals_omega0} with the given $i$ that hold as equalities; denote $n(\pi)=\sum_{i\in N}n_i(\pi)$. Hence, the extreme $\pi$ have at most $n(\pi)+1$ points in the support. We conclude that
	\begin{equation*}
	\val[B]=p(\omega_0)\cdot \sup_{{\footnotesize \begin{array}{c} \pi\in \Delta\big(\Delta(\Omega)\times\ldots\times \Delta(\Omega)\big)\\ \mbox{such that the marginals satisfy \eqref{eq_admissible_marginals_omega0} and}\\
			\big|\supp[\pi]\big|\leq n(\pi)+1
			\end{array}}	
	}\int_{\Delta(\Omega)\times\ldots\times \Delta(\Omega)} v(x_1,\ldots,x_n)\,\dd\pi(x_1,\ldots,x_n).
	\end{equation*}

Let us now discuss how many signals we need to generate an extreme $\pi\in \F_n^{\omega_0}(p)$. Using the construction from the proof of Lemma~\ref{lm_projection_admissible}, we obtain feasible one-agent marginals $(\lambda_i^\omega)_{\omega\in\Omega}$, $i\in N$, such that $\lambda_i^{\omega_0}=\pi_i$. Note that the union of supports of $\lambda_i^\omega$ over $\omega\in\Omega$ may be larger than the support of $\pi_i$ since we add $|\Omega|-1-n_i(\pi)$ point masses in \eqref{eq_unconditional_from_omega_0}. Let $(\mu^\omega)_{\omega\in\Omega}$ be a feasible family of distributions with $\mu^{\omega_0}=\pi$ and marginals $(\lambda_i^\omega)_{i\in N, \omega\in\Omega}$; for example, one can take $\mu^\omega$ to be the product of its marginals for $\omega\ne\omega_0$. By the revelation principle, any feasible family $(\mu^\omega)_{\omega\in\Omega}$ can be induced by an information structure with $|S_i|=\big|\supp[\mu_i]\big|$, where $\mu_i=\sum_{\omega_\in\Omega}p(\omega)\mu_i^\omega$; see the proof of Theorem~\ref{th_feasibility}. Thus there exists an information structure inducing $\pi$ that uses 
$$\big|\supp [\pi_i]\big|+|\Omega|-1-n_i(\pi)\leq |N|(|\Omega|-1)$$
signals per receiver.
\end{proof}

\section{Proofs for Supermodular Persuasion}\label{app_supermodular_proof}
\ed{
\begin{proof}[Proof of Lemma~\ref{lm_supermodular}]
Consider an agent-symmetric supermodular persuasion problem $B$ with sender's utility
\[
v^\omega(x_1,\ldots, x_n) = G^\omega\big(a^\omega(x_1),\ldots, a^\omega(x_n)\big).
\]
Here, each $G^\omega$ is supermodular in its $n$ arguments, symmetric with respect to permutations of these arguments, and $a^\omega:\Delta(\Omega)\to\R$ are arbitrary functions. Recall that  $v^\omega$ is also assumed to be upper semicontinuous throughout the paper to ensure the existence of an optimal solution. 

By~\eqref{eq_value_as_sup_over_feasible}, the value of $B$ can be written as maximization over feasible conditional distributions 
\[
\val[B] = \max_{\text{feasible }(\mu^\omega)_{\omega\in \Omega}}\ \sum_{\omega \in \Omega} p(\omega)\cdot \int v^\omega(x_1,\ldots,x_n)\, d\mu^\omega(x_1,\ldots,x_n).
\]
Since $B$ is agent-symmetric, permuting the coordinates $x_1,\ldots, x_n$ does not change $v^\omega$. For any feasible collection $(\mu^\omega)_{\omega\in \Omega}$ and any permutation $\sigma$ of $N$, applying $\sigma$ to coordinates yields another feasible collection $(\mu^\omega \circ \sigma)_{\omega\in \Omega}$ with the same value. Averaging over all permutations, we obtain a feasible collection $(\nu^\omega)_{\omega\in \Omega}$ where each receiver has identical marginal distributions and that achieves the same value. Therefore, we can focus on feasible conditional distributions $(\nu^\omega)_{\omega\in \Omega}$ for which all receivers share the same marginal distributions $(\lambda^\omega)_{\omega\in \Omega}$. Combining this insight with Corollary~\ref{cor_persuasion_as_transport}, we obtain the following representation
\[
\val[B]=\max_{\text{1-agent feasible }(\lambda^\omega)_{\omega\in \Omega}}\ \sum_{\omega \in \Omega} p(\omega)\cdot MK_{v^\omega}\big[\lambda^\omega,\ldots,\lambda^\omega\big].
\]
The transportation problems on the right-hand side admit a closed-form solution. Let $\tau^\omega$ be the distribution of $a^\omega(x)$ when $x$ is drawn from $\lambda^\omega$. Hence,
\begin{equation}\label{eq_diagonal_transport}
MK_{v^\omega}[\lambda^\omega,\ldots,\lambda^\omega] = MK_{G^\omega}[\tau^\omega,\ldots,\tau^\omega].
\end{equation}
From \citep[Theorem~1]{burchard2006rearrangement} for supermodular objectives, the optimal distribution in the right-hand side of \eqref{eq_diagonal_transport} is assortative matching, which, due to identical marginals, is supported on the diagonal. Hence,
\[
MK_{G^\omega}[\tau^\omega,\ldots,\tau^\omega] = \int_\R G^\omega(z,\ldots,z)\, d\tau^\omega(z).
\]
Since $x \mapsto a^\omega(x)$ transforms belief distributions $\lambda^\omega$ into $\tau^\omega$,
\[
MK_{v^\omega}[\lambda^\omega,\ldots,\lambda^\omega] = \int_{\Delta(\Omega)} G^\omega\big(a^\omega(x),\ldots,a^\omega(x)\big) \, d\lambda^\omega(x).
\]
Given that each $\lambda^\omega$ represents a feasible conditional distribution, we have $\lambda^\omega(x) = \frac{x(\omega)}{p(\omega)}\lambda(x)$, where $\lambda$ is an unconditional distribution on $\Delta(\Omega)$. Thus,
\[
MK_{v^\omega}[\lambda^\omega,\ldots,\lambda^\omega] = \int_{\Delta(\Omega)} G^\omega\big(a^\omega(x),\ldots,a^\omega(x)\big)\,\frac{x(\omega)}{p(\omega)}\, d\lambda(x).
\]
Combining these results, we obtain:
\begin{align*}
\val[B] = \max_{\text{1-agent feasible }\lambda}\ \int_{\Delta(\Omega)}\left[\sum_{\omega\in\Omega} x(\omega)\cdot G^\omega\big(a^\omega(x),\ldots,a^\omega(x)\big)\right]\, d\lambda(x)= \val[\overline{B}],
\end{align*}
where $\overline{B}$ is a single-receiver persuasion problem with a state-independent utility function
\[
\overline{v}(x)=\sum_{\omega\in\Omega}x(\omega)\, G^\omega\big(a^\omega(x),\ldots, a^\omega(x)\big).
\]
The value $\val[\overline{B}]$ is equal to the concavification of $\overline{v}$ evaluated at the prior $p$, completing the proof.
\end{proof}
}

\section{Solving the Dual Problem}\label{app_dual}

\ed{Proposition \ref{prop_dual value representation} establishes the dual to a multi-receiver persuasion problem. In this section, we present a heuristic approach to constructing explicit solutions to this dual. We illustrate the approach by finding an easy-to-check sufficient condition for the optimality of a full-information/partial-information policy. 
}  

A \emph{full-information/partial-information policy} is an information structure revealing 
the state to one receiver and partially informing the other.
Such information structures can be implemented in the model of sequential persuasion by~\cite*{khantadze2021persuasion}, where information is revealed to agents sequentially so that each next agent observes all predecessors' signals. Hence, a sufficient condition for the optimality of full-information/partial-information policy is also sufficient for the optimality of sequential persuasion.
\smallskip

 The heuristic that we rely on is that, in problems where it is optimal to fully inform one receiver, the
 solution to the dual problem is determined by the values of the utility function on the boundary of $[0,1]^2$. Relying on this intuition, we construct candidates for optimal $\alpha_i$ and $V^\omega$. The requirement that this candidate solution is indeed a solution gives a sufficient condition for the optimality of
 full-information/partial-information policy. We first illustrate the ideas for a family of sender's objectives depending on the difference of induced beliefs and formulate the condition of
 optimality for a full-information/no-information policy.
 We then extend the result to general objectives and general full-information/partial-information policies.

\smallskip 

Consider a persuasion problem with two receivers, binary state, symmetric prior $p={1}/{2}$, and $$v^h(x_1,x_2)=v^l(x_1,x_2)=h(|x_1-x_2|)$$ with some non-decreasing continuous function $h$. We aim to find a condition on $h$ so that revealing the state to one of the agents and giving no information to the other one is optimal. 

This full-information/no-information policy guarantees a payoff of $h(1/2)$ in both states. By~\eqref{eq:dual_value_2_receivers_simplified}, this payoff is optimal if  and only if there exists a function $\alpha$ such that 
\begin{eqnarray}\label{eq_alpha_conditions_h1}
h(|x_1-x_2|)&\leq& h(1/2)+(1-x_1)\alpha(x_1)+ (1-x_2)\alpha(x_2),\\
		h(|x_1-x_2|)&\leq& h(1/2)-x_1\cdot\alpha(x_1) -x_2\cdot\alpha(x_2)	
  \label{eq_alpha_conditions_h2}
\end{eqnarray}
To gain an intuition about the existence of $\alpha$, we plug $x_2=1$ into the first inequality and $x_2=0$ into the second and get 
\begin{equation}\label{eq_upper_lower_alpha_h}
    \frac{h(1-x_1)-h(1/2)}{1-x_1}\leq \alpha(x_1)\leq \frac{h(1/2)-h(x_1)}{x_1}.
\end{equation}
Hence, for $\alpha$ to exist, the left-hand side of~\eqref{eq_upper_lower_alpha_h} has to be upper-bounded by the right-hand side. Equivalently,
\begin{equation}
\label{eq_h_bar}
x_1 h(1-x_1)+(1-x_1)h(x_1)\leq h(1/2)
\end{equation}
is necessary for the optimality of a full-information/no-information policy. This condition becomes intuitive if we rewrite it as
\begin{equation}
\label{eq_h_bar_cav}
\cav[\overline{h}](1/2)\leq \overline{h}(1/2),\qquad \mbox{where}\quad 
\overline{h}(x_1)=x_1 h(1-x_1)+(1-x_1)h(x_1).
\end{equation}
Indeed, it means that in the single-receiver persuasion problem obtained from $B$ by revealing the state to the second agent, revealing no information to the first is optimal.

Assuming that~\eqref{eq_h_bar_cav} is satisfied, we find a sufficient condition for optimality of full-information/no-information policy. By~\eqref{eq_h_bar_cav}, we know that there are functions $\alpha$ satisfying~\eqref{eq_upper_lower_alpha_h}
and we select a particular one:
\begin{equation}\label{eq_alpha_h}
\alpha(x_1)=\left\{\begin{array}{cc}
   \frac{h(1-x_1)-h(1/2)}{1-x_1}, & x_1\leq 1/2\\
\frac{h(1/2)-h(x_1)}{x_1}, & x_1\geq 1/2
\end{array}\right..
\end{equation}
The idea is that we want $\alpha$ to be given the most demanding constraint, e.g., for small $x_1$, the upper bound is unlikely to be active thanks to $x_1$ in the denominator.
Plugging in this $\alpha$ into~(\ref{eq_alpha_conditions_h1}-\ref{eq_alpha_conditions_h2}), we obtain the following result.

\begin{figure}[h!]
	\begin{center}
	
\begin{tikzpicture}
\definecolor{darkgreen}{RGB}{0,100,0}
\begin{axis}[
    axis lines = middle,
    xmin=0, xmax=1,
    ymin=-1, ymax=1,
    width=7cm,
    height=7cm,
    xlabel={$x_1$},
    xtick={0.5},
    ytick={},
    yticklabels={},
    xticklabels={$\phantom{aa}1/2$},
    x tick label style={above},
    grid style={dashed},
]
]

 \draw[fill=black!50] (axis cs:0.5,-0.0) circle (3pt);
\node at (0.32,0.9) {\small$\frac{h(1/2)-h(x_1)}{x_1}$};
\node at (0.6,-0.8) {\small$\frac{h(1-x_1))-h(1/2)}{1-x_1}$};
\node at (0.1,0.4) {\small$\alpha(x_1)$};


\addplot[
    domain=0:0.9,
    samples=100,
    thick,
    gray,
] {-0.03+((1-x)*(1-x)*(1-x)-1/8)/(1-x)};

\addplot[
    domain=0.1:1,
    samples=100,
    thick,
    gray,
] {0.03+(1/8-x*x*x)/x};

\addplot[
    domain=0:0.5,
    samples=100,
    ultra thick,
    darkgreen,
] {((1-x)*(1-x)*(1-x)-1/8)/(1-x)};
\addplot[
    domain=0.5:1,
    samples=100,
    ultra thick,
    darkgreen,
] {(1/8-x*x*x)/x};
 \end{axis}

\end{tikzpicture}
  
	\end{center}
	\caption{The construction of $\alpha$ from~\eqref{eq_alpha_h} for $v(x_1,x_2)=|x_1-x_2|^3$; see Example~\ref{ex_x_to_third}. \label{fig_cubic}}
\end{figure}

\begin{proposition}\label{prop_sufficient_full_info_no_info_h}
Consider a persuasion problem with two receivers, binary state, prior $p=1/2$, and
 utility function $v^h(x_1,x_2)=v^l(x_1,x_2)=h(|x_1-x_2|)$. If $h$ is non-decreasing and satisfies the following conditions
\begin{eqnarray}\label{eq_h_conditions1}
h(x_2-x_1)\leq& h(1-x_1)+\frac{1-x_2}{x_2}(h(1/2)-h(x_2)), & x_1 \leq 1/2\leq x_2\\
h(x_2-x_1)\leq& h(1/2)-\sum_{i\in\{1,2\}}\frac{x_i}{1-x_i}\left(h(1-x_i)-h(1/2)\right), & x_1\leq x_2\leq 1/2	
\label{eq_h_conditions2}
\end{eqnarray}
then the full-information/no-information policy is optimal.
\end{proposition}
\begin{proof}
A payoff of $h(1/2)$ is guaranteed by revealing the state to the second receiver and keeping the first one uninformed. To show the optimality of this full-information/no-information policy, we need to demonstrate that the value of the persuasion problem is at most $h(1/2)$. For this purpose, it is enough to demonstrate that~(\ref{eq_alpha_conditions_h1}-\ref{eq_alpha_conditions_h2}) have a solution $\alpha$. We will show that, under the assumptions of the proposition, $\alpha$ given by~\eqref{eq_alpha_h} solves~(\ref{eq_alpha_conditions_h1}-\ref{eq_alpha_conditions_h2}). 

We need to check each of the two inequalities (\ref{eq_alpha_conditions_h1}-\ref{eq_alpha_conditions_h2}) in each of the four regions 
determined by whether $x_i$ in $[0,1/2]$ or $[1/2,1]$, $i=1,2$. Thanks to the symmetry of the problem, all these eight cases reduce to three. In  $[0,1/2]\times [1/2,1]\cup [1/2,1]\times[0,1/2]$, both inequalities (\ref{eq_alpha_conditions_h1}-\ref{eq_alpha_conditions_h2})  are equivalent to~\eqref{eq_h_conditions1} and thus hold. In $[1/2,1]^2$, inequalities (\ref{eq_alpha_conditions_h1}-\ref{eq_alpha_conditions_h2})  follow from those in $[0,1/2]^2$. Hence, it remains to verify (\ref{eq_alpha_conditions_h1}-\ref{eq_alpha_conditions_h2}) in  $[0,1/2]^2$. There, \eqref{eq_alpha_conditions_h2} holds since it reduces to~\eqref{eq_h_conditions2}. Finally, \eqref{eq_alpha_conditions_h1} in $[0,1/2]^2$ reduces to
$$h(x_2-x_1)\leq h(1-x_1)+h(1-x_2)-h(1/2),\qquad x_1\leq x_2\leq 1/2,$$
which holds trivially by the monotonicity of $h$ since $x_2-x_1$ and $1/2$ are smaller than $1-x_1$ and $1-x_2$. We conclude that (\ref{eq_alpha_conditions_h1}-\ref{eq_alpha_conditions_h2}) hold in $[0,1]^2$, thus full-information/no-information policy is optimal.
\end{proof}

Note that the conditions of Proposition~\ref{prop_sufficient_full_info_no_info_h} are formulated in terms of primitives of the model and so can be checked by an elementary (but sometimes tedious) computation.

Proposition~\ref{prop_sufficient_full_info_no_info_h} provides a useful tool for determining the optimality of the full-information/no-information policy. We develop below an alternative tool that might prove the optimality of the full-information/{partial}-information policies. The tool below is applicable to every prior $p\in (0,1)$.
For simplicity, we will keep the assumption that the problem is agent-symmetric, which allows us to focus on one function $\alpha$ instead of a pair, but this assumption could also be easily dropped. 

Consider a two-receiver persuasion problem, a binary state with prior $p\in (0,1)$, and a continuous state-dependent utility $v^\omega(x_1,x_2)=v^\omega(x_2,x_1)$. 

Suppose the sender uses a full-information/partial-information policy revealing the state to the second agent. Deciding what information to reveal to the first one reduces to solving a single-receiver persuasion problem with the sender's utility function
$$\overline{v}(x_1)=x_1\cdot v^\ell(x_1,1)+(1-x_1)v^h(x_1,0).$$
Thus, full-information/partial-information policy is optimal if and only the value of the persuasion problem does not exceed $\cav[\overline{v}](p)$. 
By~\eqref{eq:dual_value_2_receivers_simplified}, the value does not exceed $\cav[\overline{v}](p)$ if and only if there exists a function $\alpha$ such that
\begin{equation}\label{eq_alpha_conditions}
\begin{array}{ccl}
v^\ell(x_1,x_2)&\leq& V_p^\ell+(1-x_1)\alpha(x_1)+ (1-x_2)\alpha(x_2),\\
		v^h(x_1,x_2)&\leq& V_p^h-x_1\cdot\alpha(x_1) -x_2\cdot\alpha(x_2)
\end{array},			
\end{equation}
where $V_p^\ell$ and $V_p^h$ are such  that
\begin{equation}\label{eq_Vl_Vh}
x_1\cdot V_p^\ell +(1-x_1)V_p^h \quad \mbox{is the tangent line to the graph of}\quad \cav[\overline{v}]\quad \mbox{at} \ x_1=p.
\end{equation}
Note that if $\cav[\overline{v}]$ is differentiable at $x_1=p$, then $V_p^\ell=\cav[\overline{v}](p)+(1-p)\frac{d}{dx_1}\cav[\overline{v}](p)$ and $V_p^h=\cav[\overline{v}](p)-p\frac{d}{dx_1}\cav[\overline{v}](p)$.

To find a sufficient condition for the optimality of the full-information/partial-information policy, we select a particular function $\alpha$ using a  heuristic similar to the one used for the full-information/no-information policy. Plugging $x_2=1$ into the first inequality of~\eqref{eq_alpha_conditions} and $x_2=0$ to the second, we see that \begin{equation}\label{eq_upper_lower_alpha}
    \frac{v^\ell(x_1,1)-V_p^\ell}{1-x_1}\leq \alpha(x_1)\leq \frac{V_p^h-v^h(x_1,0)}{x_1}.
\end{equation}
The condition~\eqref{eq_Vl_Vh} guarantees that the graph of the function on the left-hand side in~\eqref{eq_upper_lower_alpha} lies below that of the right-hand side. Moreover, the two graphs touch each other at $x_1$ where the linear function $x_1\cdot V_p^\ell +(1-x_1)V_p^h$ touches by $\overline{v}$.
Let $b_p$ and $c_p$ be the leftmost and the rightmost such points, respectively.

We define $\alpha_p$ as follows:
\begin{equation}\label{eq_alpha_p}
\alpha_p(x_1)=\left\{\begin{array}{cc}
   \frac{v^\ell(x_1,1)-V_p^\ell}{1-x_1}, & x_1\leq b_p\\
   v^\ell(x_1,1)-V_p^\ell+V_p^h-v^h(x_1,0), & x_1\in[b_p,c_p]\\
\frac{V_p^h-v^h(x_1,0)}{x_1}, & x_1\geq c_p
\end{array}\right..
\end{equation}
In other words, for small values of $x_1$, the function $\alpha_p$ is given by the lower bound in \eqref{eq_upper_lower_alpha}, for high values $x_1$ it is given by the upper bound, and, between the points $b_p$ and $c_p$ (at these points the two bounds coincide), $\alpha_p$ equals the convex combination of the two bounds with weights $(1-x_1)$ and $x_1$. The intuition is again that  $\alpha_p$ must equal the most demanding of the bounds. 

\begin{figure}[h!]
	\begin{center}
	
\begin{tikzpicture}
\definecolor{darkgreen}{RGB}{0,100,0}
\begin{axis}[
    axis lines = middle,
    xmin=0, xmax=1,
    ymin=-1.1, ymax=1,
    width=7cm,
    height=7cm,
    xlabel={$x_1$},
    xtick={0,0.3333333,1},
    ytick={-0.3333333},
    yticklabels={},
    xticklabels={$0$, $\phantom{aa}p$, $1$},
    grid style={dashed},
]
]

\draw[thin, densely dashed] (0.333333,0)--(0.333333,-0.333333);
 \draw[fill=black!50] (axis cs:0.3333333,-0.333333) circle (3pt);
\node at (0.32,0.86) {\small$\frac{V_p^h-v^h(x_1,0)}{x_1}$};
\node at (0.33,-0.9) {\small$\frac{v^\ell(x_1,1)-V_p^\ell}{1-x_1}$};
\node at (0.9,-0.89) {\small$\alpha(x_1)$};


\addplot[
    domain=0:0.9,
    samples=100,
    thick,
    gray,
] {-0.03+(1/9-x)/(1-x)};

\addplot[
    domain=0.1:1,
    samples=100,
    thick,
    gray,
] {0.03+(2/9-x)/x};

\addplot[
    domain=0:0.33333333,
    samples=100,
    ultra thick,
    darkgreen,
] {(1/9-x)/(1-x)};
\addplot[
    domain=0.3333333:1,
    samples=100,
    ultra thick,
    darkgreen,
] {(2/9-x)/x};
 \end{axis}

\end{tikzpicture}
\hskip 2cm
\begin{tikzpicture}
\definecolor{darkgreen}{RGB}{0,100,0}
\begin{axis}[
    axis lines = middle,
    xmin=0, xmax=1,
    ymin=-0.25, ymax=0.2,
    width=7cm,
    height=7cm,
    xlabel={$x_1$},
    xtick={0.2113248,0.5,0.788675,1},
    yticklabels={},
    xticklabels={$b_p$, $p$, $\phantom{aa}c_p$,$1$},
    grid style={dashed},
]
]

\draw[thin, densely dashed] (0.2113248,0)--(0.2113248,0.083333333);
\draw[thin, densely dashed] (0.788675,0)--(0.788675,-0.083333333);
 \draw[fill=black!50] (axis cs:0.2113248,0.083333333) circle (3pt);
  \draw[fill=black!50] (axis cs:0.788675,-0.083333333) circle (3pt);
\node at (0.35,0.15) {\small$\frac{V_p^h-v^h(x_1,0)}{x_1}$};
\node at (0.7,-0.2) {\small$\frac{v^\ell(x_1,1)-V_p^\ell}{1-x_1}$};
\node at (0.1,0.55) {\small$\alpha(x_1)$};


\addplot[
    domain=0:0.5,
    samples=100,
    thick,
    gray,
] {-0.01+(0.5*(1-x)*(0.5-x)-0.0481125)/((1-x))};
\addplot[
    domain=0.5:0.95,
    samples=100,
    thick,
    gray,
] {-0.01+(0.5*(1-x)*(x-0.5)-0.0481125)/((1-x))};

\addplot[
    domain=0.5:1,
    samples=100,
    thick,
    gray,
] {0.01+(0.0481125-0.5*x*(x-0.5))/(x)};
\addplot[
    domain=0:0.5,
    samples=100,
    thick,
    gray,
] {0.01+(0.0481125-0.5*x*(0.5-x))/(x)};

\addplot[
    domain=0:0.2113248,
    samples=100,
    ultra thick,
    darkgreen,
] {(0.5*(1-x)*(0.5-x)-0.0481125)/((1-x))};
\addplot[
    domain=0.788675:1,
    samples=100,
    ultra thick,
    darkgreen,
] {(0.0481125-0.5*x*(x-0.5))/(x)};
\addplot[
    domain=0.2113248:0.5,
    samples=100,
    ultra thick,
    darkgreen,
] {0.5*(1-2*x)*(0.5-x)};
\addplot[
    domain=0.5:0.788675,
    samples=100,
    ultra thick,
    darkgreen,
] {0.5*(1-2*x)*(x-0.5)};

 \end{axis}

\end{tikzpicture}

	\end{center}
	\caption{ Construction of $\alpha_p$ from~\eqref{eq_alpha_p}.
 \textbf{Left:} $v(x_1,x_2)=|x_1-x_2|$ with $p=1/3$; full-information/no-information is optimal and thus $b_p=c_p=p$ (Example~\ref{retailer}) \textbf{Right:} $v(x_1,x_2)=|x_1-x_2|\cdot \left|x_1-{1}/{2}\right|\cdot \left|x_2-{1}/{2}\right|$ with $p=1/2$; full-information/partial information with beliefs $b_p=(3-\sqrt{3})/6$ and $c_p=1-b_p$ of the partially-informed receiver  is optimal (Example~\ref{ex_partial}).\label{fig_alpha_concavification}}
\end{figure}

\begin{proposition}\label{prop_sufficient_full_info_partial_info_optimality}
If $\alpha_p$, $V_p^\ell$, and $V_p^h$ defined by~\eqref{eq_alpha_p} and~\eqref{eq_Vl_Vh} satisfy the inequalities \eqref{eq_alpha_conditions}, then the value of the persuasion problem equals 
$\cav[\overline{v}](p)$
and a full-information/partial-information policy revealing the state to receiver $2$ and inducing the beliefs $b_p$ or $c_p$ of the first receiver is optimal.
\end{proposition}
\begin{proof}
The sender guarantees a payoff of $\cav[\overline{v}](p)$ by the information structure from the statement of the lemma. It remains to show that the sender cannot improve upon this utility level. Substituting $\alpha_p$, $V_p^\ell$, and $V_p^h$ into the dual representation for the value~\eqref{eq:dual_value_2_receivers_simplified}, we see that the value is bounded from above by $ p\cdot V_p^\ell+(1-p)V_p^h=\cav[\overline{v}](p)$. Thus the full-information/partial-information policy is optimal.
\end{proof}
Checking the conditions of Proposition~\ref{prop_sufficient_full_info_partial_info_optimality} for given sender's utility $v$
reduces to verifying inequalities between explicitly given functions on the unit square. 
{In addition to Example~\ref{retailer} from Section~\ref{sect_polarization}, we provide an example where the solution is a full-information/partial information policy with non-trivial partial information.}
\begin{example}[discord with informative signals]\label{ex_partial}
 Consider a persuasion problem with 
 $$v(x_1,x_2)=|x_1-x_2|\cdot \left|x_1-\frac{1}{2}\right|\cdot \left|x_2-\frac{1}{2}\right|$$  
 and prior $p=1/2$. Here  
the sender is incentivized to push induced beliefs further away from each other and also from $1/2$, i.e., the sender aims to induce discord while keeping both agents' signals informative. In particular, revealing no information to one of the agents is definitely suboptimal. Indeed, the persuasion problem 
 satisfies the conditions of Proposition~\ref{prop_sufficient_full_info_partial_info_optimality} with $b_p=(3-\sqrt{3})/6= 0.211...$ and $c_p=1-b_p$; see Figure~\ref{fig_alpha_concavification}:right and Mathematica code in Appendix~\ref{app_code2}. Thus full-information/partial information policy inducing beliefs $b_p$ and $c_p$ of the less informed receiver is optimal.
 \end{example}

\section{Code for Example~\ref{ex_x_to_third}}\label{app_code}
The following Mathematica code finds the maximal $\beta$ such that the inequalities from Proposition~\ref{prop_sufficient_full_info_no_info_h} are satisfied for $h(|t|)=|t|^\beta$. The algorithm implements a binary search with respect to $\beta$ and outputs $\beta=2.25751...$
\begin{verbatim}
ClearAll; 
h[t_, beta_] := Abs[t]^beta; (*define function h*)
(*define the difference between the LHS and the RHS side of the inequalities to be checked*)
ineq1[x_, y_, beta_] := h[x - y, beta] - (h[1 - x, beta] 
  + (1 - y)/y *(h[0.5, beta] - h[y, beta]));
ineq2[x_, y_, beta_] := h[x - y, beta] - (h[0.5, beta] 
  -  x/(1 - x) *(h[1 - x, beta] - h[0.5, beta]) -  y/(1 - y) *(h[1 - y, beta] - h[0.5, beta]));
(*define the precision and the range for beta*)
betaPrecision = 10^-6;
betaMin = 0;
betaMax = 100;
(*binary search for the maximal beta*)
While[betaMax - betaMin > betaPrecision,
  beta = (betaMin + betaMax)/2;
  If[
   NMaximize[{ineq1[x, y, beta], 0 <= x <= 0.5, 0.5 <= y <= 1}, {x, y}][[1]] <= $MachineEpsilon 
        && (*if both inequalities hold within machine precision for all x and y*) 
   NMaximize[{ineq2[x, y, beta], 0 <= x <= 0.5, 0 <= y <= 0.5}, {x, y}][[1]] <= $MachineEpsilon,
   betaMin = beta, (*then increase betaMin*)
   betaMax = beta  (*else decrease betaMax*)
   ] 
];
beta (*print beta*)
\end{verbatim}

\section{Code for Example~\ref{ex_partial}}\label{app_code2}
Consider any utility function $v(x_1,x_2)$ such that $v(x_1,x_2)=v(x_2,x_1)$ and $v(x_1,x_2)=v(1-x_1,1-x_2)$. The following Mathematica code checks that $v$ satisfies the conditions of Proposition~\ref{prop_sufficient_full_info_partial_info_optimality}.

\begin{verbatim}
ClearAll;
v[x_, y_] := Max[0, Abs[x - y] Abs[x - 1/2] Abs[y - 1/2]]; (*define the utility function*)
vBar[x_] :=  x*v[x, 1] + (1 - x)*v[x, 0]; (*define the auxiliary function \bar{v}*)
(*by the symmetry of v, the global maximum of \bar{v} equals 
the maximum of its concavification at the prior 1/2*)
maxPoint = Maximize[{vBar[x], 0 <= x <= 1}, x]; 
V = maxPoint[[1]]  (*V is the maximum*)
(*b and c are optimal posteriors of the partially informed receiver*)
b = Min[x /. maxPoint[[2]], 1 - x /. maxPoint[[2]]]; 
c = 1 - b;
(*define function alpha*)
alpha[x_] :=  Piecewise[{
    {(v[x, 1] - V)/(1 - x),    0 <= x < b},
    {v[x, 1] - v[x, 0],        b <= x < c},
    {(V - v[x, 0])/x,          c <= x <= 1}
}];
(*by symmetry, it is enough to check only one inequality from the proposition; 
define ineq as the difference between the LHS and the RHS*)
ineq[x_, y_] := v[x, y] - (V + (1 - x)*alpha[x] + (1 - y)*alpha[y]);
(*if the difference is non-positive within precision, 
the conditions of the proposition are satisfied*)
If[NMaximize[{ineq[x, y], 0 <= x <= 1, 0 <= y <= 1}, {x, y}][[1]] <= $MachineEpsilon, 
  Print["Full-info/partial-info is optimal"], 
  Print["Full-info/partial-info may not be optimal"]
];
\end{verbatim}

\end{document}